\documentclass[preprint, twocolumn, 3p, times]{elsarticle}
\usepackage{xcolor}

\usepackage[utf8]{inputenc}
\usepackage{amsmath,amssymb,amsfonts, amsthm}
\usepackage{mathtools, pifont, nicefrac}
\usepackage{xcolor}
\usepackage{graphicx}
\usepackage{subcaption}
\usepackage{tikz}
\usetikzlibrary{hobby,decorations.markings}
\usepackage[]{pgfplots}
\pgfplotsset{compat=newest}
\usepackage{tabularx}
\usepackage{multirow}
\usepackage{booktabs}
\usepackage[colorlinks=true]{hyperref}
\usetikzlibrary[patterns.meta, positioning, hobby, calc, arrows.meta, decorations.pathreplacing, backgrounds, matrix, calligraphy, shapes.misc]

\newcommand{\QED}{\square}

\definecolor{dgreen}{RGB}{61,152,56}

\newcommand{\new}{\textcolor{black}}

\newcommand{\cmark}{\ding{51}}%
\newcommand{\xmark}{\ding{55}}

\newcommand{\cX}{\mathcal{X}}
\newcommand{\cI}{\mathcal{X}_I}

\newcommand{\cG}{\mathcal{X}_G}
\newcommand{\cS}{\mathcal{X}_S}
\newcommand{\cU}{\mathcal{X}_U}
\newcommand{\cF}{\mathcal{X}_F}
\newcommand{\real}{\mathbb{R}}
\newcommand{\extreal}{\overline{\mathbb{R}}}

\DeclareMathOperator{\interior}{int}
\DeclareMathOperator{\Sphere}{Sphere}
\DeclareMathOperator{\Torus}{Torus}
\DeclareMathOperator{\Rectangle}{Rectangle}

\newtheorem{theorem}{\bfseries Theorem}
\newtheorem{corollary}[theorem]{\bfseries Corollary}%
\newtheorem{remark}[theorem]{\bfseries Remark}
\newtheorem{certificate}{\bfseries Certificate}

\newtheorem{example}{\bfseries Example}%

\newtheorem*{example*}{\bfseries Example}

\journal{Annual Reviews in Control}

\begin{document}

\begin{frontmatter}

\title{A General Framework for Verification and Control of Dynamical Models\\ via Certificate Synthesis}

\author[oxford]{Alec Edwards}

\affiliation[oxford]{organization={Department Of Computer Science},%
            city={Oxford},
            country={UK}}
\ead{alec.edwards@cs.ox.ac.uk}
\author[delft]{Andrea Peruffo}
\author[oxford]{Alessandro Abate}

\affiliation[delft]{organization={Center for Systems and Control},%
            city={Delft},
            country={The Netherlands}}

\begin{abstract}
    An emerging branch of control theory specialises in \emph{certificate learning}, 
    concerning the specification of a desired (possibly complex) system behaviour  
    for an autonomous or control model, 
    which is then analytically verified by means of a function-based proof. 
    However,  
    the synthesis of controllers abiding by these complex requirements is in general a non-trivial task and may elude the most expert control engineers. 
    This results in a need for automatic techniques that are able to design controllers and to analyse a wide range of elaborate specifications. 
    In this paper, we provide a general framework to encode system specifications and define corresponding certificates, and we present an automated approach to formally synthesise controllers and certificates. 
    Our approach contributes to the broad field of safe learning for control, exploiting the flexibility of neural networks to provide candidate control and certificate functions, whilst using SAT-modulo-theory (SMT)-solvers to offer a formal guarantee of correctness. 
    We test our framework by developing a prototype software tool, and
    assess its efficacy at verification via control and certificate synthesis over a large and varied suite of benchmarks. 
\end{abstract}


\end{frontmatter}

\section{Introduction}

The analysis of the behaviour of continuous-time dynamical systems focuses on a wide range of properties, which are themselves suitable for an even wider range of applications. Properties of interest include arguably the most common property, (asymptotic) stability, namely the convergence of trajectories to an equilibrium \cite{sastry1999NonlinearSystems}; to safety, namely the avoidance of an unsafe region of the state space at all time \cite{blanchini2008SetTheoreticMethodsControl}; to its dual, reachability, that is the hitting of a target region of the state space in finite time  \cite{henzinger1996TheoryHybridAutomata}. 
As we shall see, with the combination and the slight modification of these basic properties (we shall alternatively denote them as specifications or requirements), an engineer may design a broad range of desired dynamical behaviours for any model at hand. 

Given a model of a dynamical system, 
such spectrum of properties can be investigated from different perspectives and with diverse approaches: either analytical (e.g., via local linearisation and eigenvalues computation) \cite{sastry1999NonlinearSystems}, or computational ones (e.g., via dynamic flow propagation or via reach-set computation). In general, non-linearity in dynamical models is difficult to deal with: on account of this,  approaches that are \emph{indirect} or \emph{sufficient} can be successful: a proof that the system actually fulfils a given requirement can be offered in the form of a \emph{certificate}: the onus is to find, or to synthesise, a real-valued function defined over the state space with proper characteristics.  
A celebrated instance of indirect methods is the synthesis of Lyapunov functions \cite{lyapunov1992GeneralProblemStability}, whereby one ought to hand-craft a bespoke energy function, oftentimes based on intuition and on physical properties of the underlying dynamical model.

In recent years numerical optimisation methods have automated the synthesis of certificates, employing \emph{templates}, i.e. candidate functions where only the coefficients (parameters) ought to be determined. 
Commonly, the choice falls onto polynomial templates framed as sum-of-squares convex problems  \cite{papachristodoulou2002ConstructionLyapunovFunctions, papachristodoulou2013SOSTOOLSVersion00, prajna2006BarrierCertificatesNonlinear, goubault2014FindingNonPolynomialPositive},  which admit globally optimal solutions. However, these techniques operate solely on models with polynomial dynamics and various convexity assumptions.
Alternative formulations include linear programs \cite{bensassi2016LinearRelaxationsPolynomial, sankaranarayanan2013LyapunovFunctionSynthesis, ratschan2010ProvidingBasinAttraction} and
semi-algebraic systems  \cite{she2009SemialgebraicApproachAsymptotic, she2013DiscoveringPolynomialLyapunov}, 
all of which raise structural requirements on the dynamical models at hand. 

Despite the usefulness of the mentioned synthesis approaches, 
they are numerically sensitive and generally unsound, and thus undesirable for robust solutions and safety-critical applications \cite{Aba:17,DBLP:conf/pldi/BohrerTMMP18,knightSafetyCriticalSystems2002}.  Consequently, in recent years interest has grown in approaches for synthesis that can yield \emph{provably-correct} certificates, much in the same line of research as \emph{correct-by-design} control synthesis \cite{tabuada09,belta2017formal}. A powerful technique to reason formally about correctness involves SMT-solving \cite{barrett2010smt}. 
SAT-modulo-theory (SMT) extends satisfiability (SAT) solving to richer theories, enabling, for example, finding feasible assignments of real numbered variables over nonlinear formulae.

SMT can be in particular leveraged for synthesis tasks. \emph{Inductive} approaches \cite{solar-lezama2006CombinatorialSketchingFinite}, leveraging SMT, have been used to synthesise certificates \cite{ravanbakhsh2015CounterexampleGuidedSynthesisReachWS, ravanbakhsh2019LearningControlLyapunov}, controllers \cite{huang2015controller, ABCCDKKP20} and abstractions \cite{abate2022neural} for dynamical models. Such techniques have been used first for stability certification of dynamical models using polynomial Lyapunov functions and later extended to more general reach-avoid requirements \cite{ahmed2018AutomatedSoundSynthesis, kapinski2014SimulationguidedLyapunovAnalysis, ravanbakhsh2015CounterExampleGuidedSynthesis, ravanbakhsh2015CounterexampleGuidedSynthesisReachWS}.
Related to SMT-based solutions, approaches that formulate synthesis problems as a mixed-integer linear programs have also been used to synthesise provably-correct Lyapunov functions \cite{dai2020CounterexampleGuidedSynthesis, dai2021LyapunovstableNeuralnetworkControl} for stability analysis, 
encompassing linear matrix inequalities for uncertain systems \cite{masti2023counter}, 
and barrier certificates for safety \cite{zhao2021LearningSafeNeural, chen2020LearningLyapunovFunctions, chen2021LearningLyapunovFunctions}.  
Notably, mixed-integer problems also encompass  candidates in the form of neural networks with ReLU activation functions, 
and may employ an optimisation engine like Gurobi \cite{gurobi} to certify the soundness of the proposed functions \cite{zhao2021SynthesizingReLUNeural}.

\textbf{Related work} 
The flexibility of neural networks has permeated the field of certificate synthesis, including their use for synthesis of Lyapunov-like functions. 
For instance,
\cite{richards2018LyapunovNeuralNetwork,noroozi2008GenerationLyapunovFunctions, jin2020NeuralCertificatesSafe} describe generally unsound procedures for
gradient descent-based training of a Lyapunov neural network.. 
Sound, 
counter-example based techniques are proposed in e.g. \cite{ahmed2018AutomatedSoundSynthesis, chang2020NeuralLyapunovControl, abate2020AutomatedFormalSynthesis, samanipour2023StabilityAnalysisController, grande2023augmented, grande2023passsive}, specifically for Lyapunov functions, and solely for barrier certificates in e.g. \cite{peruffo2021AutomatedFormalSynthesis, zhao2020SynthesizingBarrierCertificates, ratschan2018SimulationBasedComputation}.
The choice of SMT solver depends on the models under consideration and desired certificate template: Z3 \cite{demoura2008Z3EfficientSMT} handles polynomial functions, dReal \cite{gao2013dRealSMTSolver}, iSat3 \cite{isat3} and CVC5 \cite{DBLP:conf/tacas/BarbosaBBKLMMMN22} enable analysis of non-polynomial functions as well as polynomials.   %
The synthesis of certificates includes more complex properties, as proposed in \cite{verdier2020FormalControllerSynthesis, verdier2020FormalSynthesisAnalytic} where bespoke genetic algorithms are leveraged to generate \emph{reach-while-stay} and \emph{reach-and-stay-while-stay} functions, alongside controllers, for hybrid systems. Certificates for reach-avoid properties have been previously synthesised using counterexample-based approaches \cite{ravanbakhsh2015CounterexampleGuidedSynthesisReachWS}. Meanwhile reach-avoid-stay properties and corresponding certificates have been well studied from a theoretical perspective \cite{meng2021ControlNonlinearSystems, meng2021SmoothConverseLyapunovBarrier}. In this work, we collate certificates for these more complex properties, and categorise them in order to unify them within simpler certificates for stability, reachability and safety. We also generalise the concept of reach-avoid-stay property by separating allowing the \emph{stay} set to be different from the \emph{reach} set.
The interested reader may find a survey on neural certificates with application in control synthesis and robotics in \cite{dawson2022safe}. 

\textbf{Contributions} 
We summarise our contributions as follows:
\begin{itemize}
    \item We collate and add to existing certificates across literature for nonlinear continuous-time dynamical models.
    \item We categorise the properties these certify into a simplified and general framework, which we newly describe through notions \emph{arrive, avoid} and \emph{remain}. 
    \item We describe a unified algorithm to concurrently synthesise \emph{both} controllers and certificates in parallel for dynamical models, to prove they satisfy these properties.
    \item We implement our framework\footnote{Avaliable at \href{https://github.com/oxford-oxcav/fossil}{https://github.com/oxford-oxcav/fossil}} on top of the computational library \textsf{Fossil} \cite{abate2021FOSSILSoftwareTool, fossil2}, offering a new prototype software tool that can verify general \emph{arrive, avoid} and \emph{remain} properties for nonlinear control models using controllers and certificates based on neural networks.
\end{itemize}

\textbf{Organisation}
This manuscript is organised as follows. Section \ref{sec:prelim} provides relevant background information and notation used across this work. In Section \ref{sec:properties}, we describe a range of properties for dynamical models and certificates which prove that they hold. Next, we describe a unified and computationally correct algorithm for synthesising these certificates in Section \ref{sec:synthesis}. We present experimental results for a prototype tool to synthesise these certificates in Section \ref{sec:experiments}, before discussing the limitations of our framework in Section \ref{sec:lim} and providing concluding remarks in Section \ref{sec:conclusion}. Finally, we outline the proofs of all theorems in \ref{app:proofs}, and 
we reserve a discussion on the broader taxonomic connections of our work to \ref{sec:connections}.

\section{Preliminaries}
\label{sec:prelim}

\subsection{Dynamical Models}

We denote the set of positive real numbers and its extended version 
as  $\real_+$ and $\extreal = \real_+ \cup \{+\infty\}$, respectively. 
A function is said to be of class $\mathcal{C}^1$ if its first derivative exists and is continuous.
Let us consider models described by 
\begin{equation}
    \label{eq:con-model}
    \dot{\xi}(t) = 
    f_u(\xi(t), u(t)), 
    \quad 
    x(t_0) = x_0 \in \cI \subseteq \cX
\end{equation}
where $x=\xi(t) \in \cX \subseteq \mathbb{R}^n$ is the state of the system,
$f_u : \cX \times \mathcal{U} \rightarrow \mathbb{R}^n$ is a Lipschitz-continuous vector field describing the model dynamics. We refer to these models as control models.
We denote a trajectory over a time horizon $T \in \extreal$ 
as $\xi(t): [t_0, T] \to \real^n$, where $\xi(t)$ admits a time derivative everywhere, and such that $\dot{\xi}(t) = f_u(\xi(t), u(t))$ and $\xi(t_0) \in \cX$, namely the trajectory is a solution of the model in \eqref{eq:con-model}.  
Finally, $u(t) \in \mathcal{U} \subseteq \mathbb{R}^m$ is the input and $\cX_I$ denotes the set of initial conditions. 
Once a state-feedback controller $u(t) = k(\xi(t))$ has been specified, we may interpret the dynamics described by \eqref{eq:con-model} as those of a closed-loop model, as follows 
\begin{equation}
    \label{eq:dyn-model}
    \dot{\xi}(t) = f(\xi(t)), 
    \quad 
    \xi(t_0) = \xi_0 \in \cI \subseteq \cX.
\end{equation}
We refer to models of this kind simply as autonomous dynamical models. 
We denote with $x^*$ an equilibrium point of \eqref{eq:dyn-model}, namely where $f(x^*) = 0$.

\subsection{Systems and Properties}

The goal of this work is to find a feedback controller, namely a signal $u(t)$ in time, such that 
the dynamics above satisfy some desired temporal requirements (e.g., safety), or to show that given closed-loop (autonomous) dynamics are endowed with some given property (e.g., asymptotic stability).

We interpret properties of dynamical models in \eqref{eq:dyn-model} in terms of their trajectories, and of binary relations that these trajectories have with given sets within the state space $\mathcal X$.  
To this end, we now introduce the notations and the semantics of the sets characterising properties of dynamical models.  
We denote as $\cX_U$ an unsafe set, indicating a region of the state space where the system's trajectories should avoid; 
$\cX_G$ represents a goal set, indicating the region that the system's trajectories should enter;
and $\cX_F$ represents a final set, indicating a set where the system's trajectories should remain for all times after arriving at the goal set. 
These sets will be employed in the next section for the definition of properties, as depicted in Fig.~\ref{fig:properties-example}. 
Implicitly, we consider that the unsafe and final sets are disjoint, i.e. $\cU \cap \cF = \emptyset$ and that the goal set is contained within the final set, i.e. $\cG \subset \cF$. \new{ Often we will assume these sets to be compact; the
motivation for assuming sets to be compact (rather than just closed) is mainly to
ease automated synthesis and verification.}
Additional topological properties of sets will be clarified later, within formal statements. 
Given a set $S$ in a domain $\cX$, we denote by ${S}^\complement$ its complement, i.e. $\cX \setminus S$, and by $\interior(S)$ its interior, namely the set without its border, i.e. $\interior(S) = S \setminus \partial S$. 

We consider a set $S$ to be (forward) \emph{invariant} if at some initial time $t_0$, $\xi(t_0) \in S$ implies that for all $t > t_0$, $\xi(t) \in S$ \cite{blanchini2008SetTheoreticMethodsControl}. We consider a set $S_A$ to be attracting~\cite{blanchini2008SetTheoreticMethodsControl} with \emph{region of attraction} $S_B$ if for some initial time $t_0$ and any initial state $\xi(t_0) \in S_B$, the trajectory $\xi(t)$ converges to $S_A$ as $t \to \infty$, i.e. if $\lim_{t\to \infty} dist(\xi(t), S_A) = 0$. \new{Finally, given a function $C: \real^n \to \real$, we denote the Lie derivative of this function with respect to the vector field $f$ simply as $\dot{C} = \langle \nabla C, f(x) \rangle$.}

\subsection{Neural Networks}
Denote a neural network $\mathcal{N}$  with input layer $z_0 \in \real^n$, corresponding to the dimension of the dynamical model in \eqref{eq:dyn-model}. This is followed by $k$ hidden layers $z_1, \dots, z_k$ with dimensions $h_1, \dots, h_k$ respectively,
and finally followed by an output layer $z_{k+1} \in \real^d$. In this work, $d \in \{1, m\}$, where $d=1$ corresponds to a scalar-valued certificate, whereas $d=m$ is used for a state-feedback controller with $m$ control variables.

\new{We denote the hidden layers and output layer as $h_i$ with index $i=0, \ldots, k+1$, where $k+1$ denotes the output layer}; to these layers are associated matrices of weights $W_i \in \real^{h_i \times h_{i-1}}$ and a vector of biases $b_i \in \real^{h_i}$~\cite{mackay2003InformationTheoryInference}. 
Every $i$-th hidden layer is associated with an activation function $\sigma_i \colon \mathbb{R} \to \mathbb{R}$.
The valuation of output and hidden layers is given by 
\begin{align}
\label{eq:layer}
  z_i &= \sigma_i(W_i \cdot  z_{i-1} + b_i), \ i = 1, \dots, k, \\
  z_{k+1} &= W_{k+1} \cdot z_k + b_{k+1},  
\end{align}
where each $\sigma_i$ is \new{applied element-wise to its
$h_i$-dimensional argument}.  

\subsection{SAT-Modulo Theory (SMT)}

Here, we offer a brief introduction to SMT solving to aid understanding of how they are used in this work; for a more comprehensive introduction the reader is directed to \cite{BT18}. 

Let $\phi(x)$ be a formula in first-order logic, expressed in terms of variables $x$ in a given domain. SMT is the problem of determining if there exists an assignment of the variables $x$ such that $\phi(x)$ is satisfied (i.e., if it evaluates to true). SMT extends the Boolean satisfiability (SAT) problem: SAT solvers use modern forms of the DPLL algorithm \cite{DBLP:journals/jacm/DavisP60, DBLP:journals/cacm/DavisLL62} to search for satisfying assignments for boolean formulae. As an example, consider the following boolean formula, 
$\phi_1(x) = (x_0 \vee x_1) \wedge (\neg x_0 \vee \neg x_2), x = (x_0,x_1,x_2), x_i \in \{T, F\}$. When provided with this formula, a SAT solver would return an assignment such as $x_0=\text{T}, x_1=\text{T}$, $x_2=\text{F}$, which satisfies $\phi_1(x)$. 

Now let us consider an example that is more relevant to this work.  Let $\sigma_t$ be the hyperbolic tangent function, and consider the following function $g: \real \to \real$
\begin{equation*}
    g(x) = \sigma_t(2 - 2x) - 3\sigma_t(1-x) + \frac{3}{2}.
\end{equation*}
Suppose we wish to find if $g$ is negative \new{at some point} where $x$ is non-negative. 
Formally, this is 
\begin{equation*}
    \exists x \in \real, x \geq 0 : g(x) < 0.
\end{equation*}
An SMT solver over the theory of real algebra, which can in particular handle transcendental functions, can handle this problem. This solver will return $\textit{unsat}$ to this problem, as no satisfying assignment exists - it is unsatisfiable. Crucially, SMT-solvers are \emph{sound}: if a satisfying assignment exists then it cannot return that the problem is unsatisfiable. 

We note that $g(x)$ is in fact a trivial feed-forward neural network and we have proven that it is positive in a region of the state-space. Later,  we will use the same ideas to prove similar properties for non-trivial networks.

\section{Properties and Certificates}
\label{sec:properties} 

We present a number of properties (or requirements) for dynamical models defined over their trajectories, alongside definitions of corresponding certificates, whose existence serve as \emph{sufficient conditions} for the satisfaction of the desired properties. 
We focus on continuous-time models; while the presented properties may be seamlessly applied to discrete time models, the corresponding certificates would not necessarily align with the presentation of the work, hence we omit their discussion as outside the scope of this work.
The proofs of the theorems relating certificates to dynamical properties are outlined in \ref{app:proofs}.
Several properties (and corresponding certificates) are ubiquitous across the control theory literature, whilst others have been more recently introduced or inherited from analogues in formal verification. Further, we include a practical variant on Lyapunov functions to allow for constrained local stability and an original certificate called Reach-Avoid-Remain, and suggest that more can be obtained in a modular, composable fashion. 
We emphasise that while many of these properties are similar to each other, they are subtly, and importantly, different and distinguished. Later in the section we provide a summary and clarification on their distinguishing features, which may also be seen in Fig.~\ref{fig:properties-example}.

\subsection{Stability}

Stability is the most-studied property of dynamical models, and many definitions of this property exist. Stability is most commonly characterised in a Lyapunov (asymptotic) sense \cite{sastry1999NonlinearSystems}, namely in terms of the distance of a trajectory from the equilibrium point.  Conversely, in order to achieve a consistent characterisation of properties in terms of set containment, here we characterise stability as follows: 
\begin{equation}
\label{eq:stab-spec}
        \exists \cI: \forall \xi(t_0) \in \cI, \exists T \in \extreal, \forall \tau \geq T, \ \xi(\tau) \in \{x^*\},
\end{equation}
where $\cI$ has non empty interior.
In words, there exists some initial set $\cI$ such that for all trajectories initialised in $\cI$, there exists a time instant $T$ (possibly at infinity) when the trajectory reaches the equilibrium state $x^*$ and remains there for all times after $T$. A model can be proven to satisfy this property using a Lyapunov function, which we introduce next.

\begin{certificate}[Lyapunov Function]
    Given a model $f$ with unique equilibrium point $x^* \in \interior(\cX)$, consider a function $V: \cX \subset \real^n \to \real, V \in \mathcal{C}^1$. $V$ is a Lyapunov function if:
    \begin{subequations}
    \label{eq:lyap}
    \begin{align}
        &V(x^*) = 0,  \\
        &V(x) > 0 \quad \forall x \in \cX \setminus \{x^*\}, \\
        &\dot{V}(x) = \langle \nabla V(x), f(x) \rangle < 0 \quad \forall x \in \cX \setminus \{x^*\}.
    \end{align}
    \end{subequations}
\end{certificate}
\medskip

\begin{theorem}[Stability]
\label{thm:lyap}
     Given a model \eqref{eq:dyn-model}, if a Lyapunov function exists, then \eqref{eq:stab-spec} holds for some set of initial conditions $\cI$. \hfill$\blacksquare$
\end{theorem}

\smallskip
Let us clarify the issue of finding $\cI$ in the next section.

\subsection{Region of Attraction} 
    Thm.~\ref{thm:lyap} proves the existence of some region of the state space in which initialised trajectories will converge asymptotically towards the origin -- this region is known as a \emph{region of attraction} (ROA). 
    In general, Lyapunov functions merely prove the existence of a region of attraction within the state space, without additional information about its size or shape. 
    Lyapunov functions may prove global asymptotic stability under additional conditions, but these can be difficult to synthesise and verify automatically for some models. Here, we offer a certificate the acts as a middle ground 
    between local and global asymptotic stability.
    
    Proving that all trajectories initialised within a given $\cI$ are stable amounts to proving that $\cI$ is contained wholly within a sub-level set of a Lyapunov function (which must also lie in $\cX$), and that the Lyapunov conditions in \eqref{eq:lyap} hold over this entire sub-level set.
    This is treated in the following equation and corollary.

First, we modify \eqref{eq:stab-spec} to now require a specified set of initial states $\cI$, which should be a region of attraction for an equilibrium point, as follows:
\begin{equation}
\label{eq:roa-spec}
        \forall \xi(t_0) \in \cI, \exists T \in \extreal, \forall \tau \geq T, \ \xi(\tau) \in \{x^*\}. 
\end{equation}

We can certify that an autonomous model satisfies this property using the following certificate.

\begin{certificate}[ROA Certificate]\label{cert:roa}
    Let a dynamical model $f$ be given with unique equilibrium point $x^* \in \cX$. A Lyapunov function $V$ is an ROA certificate if there exists a $\beta$ such that $\cI \subset \{x \in \cX : V(x) \leq \beta \ \}$ and that the conditions in \eqref{eq:lyap} hold over the set $\{x \in \cX : V(x) \leq \beta  \}$. \hfill$\blacksquare$
\end{certificate}

Alternatively, a Lyapunov function can be interpreted as a proof that a region of attraction (here $\cI$) exists within some larger set (here $\cX$), whereas a ROA certificate proves the converse: that a given initial set $\cI$ lies within a larger region of attraction, which is defined by a sublevel set of $V$. 
This certificate offers a more practical guarantee over classical Lyapunov functions: all trajectories initialised within the pre-defined set $\cI$ indeed converge to $x^*$. 

\begin{corollary}[Region of Attraction]
\label{cor:roa}
    Given a model \eqref{eq:dyn-model}, a bounded set of initial conditions $\cI$, and a ROA certificate, then \eqref{eq:roa-spec} holds. \hfill$\blacksquare$
\end{corollary}

\begin{figure*}[htbp]
    \centering
    \resizebox{0.7\textwidth}{!}{
         \subfloat[][Stability]{\begin{tikzpicture}[scale=0.45, inner sep=2pt, minimum size=0.4cm,node
    distance=.5cm,use Hobby shortcut,
    vecstab/.style={
    decorate,
    decoration={
      markings,      %
      mark=at position 0.4 with {\arrow[rotate=5, line width=2pt]{stealth[scale=10]}},
    }
  }
]

    \begin{scope}[xshift=0cm, yshift=0cm]
    \draw[fill opacity=0.35, fill=green!60!blue, text opacity=1] (0,0) circle (0.1cm) node[anchor=west, xshift=-0.2cm, yshift=0.2cm]{\tiny $x^*$};
    \end{scope}

    \begin{scope}[yshift=0.1cm]
    \draw[pattern={Lines[angle=45, distance=4pt, line width=1pt]}, pattern color=blue, fill opacity=0.5, text opacity=1]   plot[smooth, tension=.5] coordinates {(-1.7,0) (-1.9, -1.9) (0,-1) (1,0) (1, 1) (0,1) (-0.75, 0.75) (-1.55, 0.15) (-1.7, 0)} node[anchor=west, yshift=0.1cm, xshift=-0.5cm] {\tiny $\cX_I$};
    \end{scope}

    \draw[very thick, decoration={markings,
    mark=at position 0.5 with \arrow{stealth}},postaction=decorate] (-1, 0) .. (0.8, 0.3) .. (-0.5, -0.6) .. (0, 0);

\end{tikzpicture} \label{fig:stab}}
         \subfloat[][ROA]{\begin{tikzpicture}[scale=0.45, inner sep=2pt, minimum size=0.4cm,node
    distance=.5cm,use Hobby shortcut,
    vecstab/.style={
    decorate,
    decoration={
      markings,      %
      mark=at position 0.4 with {\arrow[rotate=5, line width=2pt]{stealth[scale=10]}},
    }
  }
]

    \begin{scope}[xshift=0.1cm, yshift=-1cm, scale=0.6, rotate=90]
    \draw[fill opacity=0.3, fill=blue, text opacity=1]   plot[smooth, tension=.5] coordinates {(-3.5,0.5) (-3,2.5) (-1,3.5) (1.5,3) (4,3.5) (5,2.5) (5,0.5) (2.5,-2) (0,-0.5) (-3,-2) (-3.5,0.5)} node[anchor=west, yshift=0.1cm, xshift=-1cm] {\tiny $\cX_I$};
    \end{scope}
    
    \begin{scope}[xshift=0cm, yshift=0cm]
    \draw[fill opacity=0.35, fill=green!60!blue, text opacity=1] (0,0) circle (0.1cm) node[anchor=west, xshift=-0.2cm, yshift=0.2cm]{\tiny $x^*$};
    \end{scope}

    \draw[very thick, decoration={markings,
    mark=at position 0.5 with \arrow{stealth}},postaction=decorate] (-1, 0) .. (0.8, 0.3) .. (-0.5, -0.6) .. (0, 0);

\end{tikzpicture} \label{fig:roa}}
         \subfloat[][Safety]{\begin{tikzpicture}[scale=0.45, inner sep=2pt, minimum size=0.4cm,node
    distance=.5cm,
    use Hobby shortcut,
    vec1/.style={
    decorate,
    decoration={
      markings,      %
      mark=at position 0.4 with {\arrow[rotate=5, line width=2pt]{stealth[scale=10]}},
    }
  }
]

    \begin{scope}[xshift=3cm, yshift=4cm, scale=0.3]
    \draw[fill opacity=0.3, fill=blue, text opacity=1]   plot[smooth, tension=.5] coordinates {(-3.5,0.5) (-3,2.5) (-1,3.5) (1.5,3) (4,3.5) (5,2.5) (5,0.5) (2.5,-2) (0,-0.5) (-3,-2) (-3.5,0.5)} node[anchor=east, yshift=0cm, xshift=0cm] {\tiny $\cX_I$};
    \end{scope}

    \begin{scope}[scale=0.3, xshift=8cm, yshift=2.5cm]
    \draw[fill opacity=0.35, fill=orange, text opacity=1] plot[smooth, tension=.5] coordinates {(-3.5,0.5) (-1,1) (-1,2) (1.5,3) (4,3.5) (5,2.5) (5,0.5) (2.5,-2) (0,0) (-3,-2) (-3.5,0.5)} node[anchor=west, xshift=0cm, yshift=-0.5cm] {\tiny $\cX_U$};
    \end{scope}

    \draw[very thick, decoration={markings,
    mark=at position 0.5 with \arrow{stealth}},postaction=decorate] (3,4) .. (1, 2) .. (0.8, 0.3) .. (-0.5, -0.6) .. (0, 0);

    \draw [very thick, decoration={markings,
    mark=at position 0.5 with \arrow{stealth}},postaction=decorate] (2.3, 4.8) .. (0.2, 4) .. (-0.5, 0.7) .. (-2, -1.5);

    \end{tikzpicture} \label{fig:safe}}
         \subfloat[][SWA]{\begin{tikzpicture}[scale=0.45, inner sep=2pt, minimum size=0.4cm,node
    distance=.5cm,
    use Hobby shortcut,
    vec1/.style={
    decorate,
    decoration={
      markings,      %
      mark=at position 0.4 with {\arrow[rotate=5, line width=2pt]{stealth[scale=10]}},
    }
  }
]

    \begin{scope}[xshift=3cm, yshift=4cm, scale=0.3]
    \draw[fill opacity=0.3, fill=blue, text opacity=1]   plot[smooth, tension=.5] coordinates {(-3.5,0.5) (-3,2.5) (-1,3.5) (1.5,3) (4,3.5) (5,2.5) (5,0.5) (2.5,-2) (0,-0.5) (-3,-2) (-3.5,0.5)} node[anchor=east, yshift=0cm, xshift=0cm] {\tiny $\cX_I$};
    \end{scope}

    \begin{scope}[scale=0.3, xshift=8cm, yshift=2.5cm]
    \draw[fill opacity=0.35, fill=orange, text opacity=1] plot[smooth, tension=.5] coordinates {(-3.5,0.5) (-1,1) (-1,2) (1.5,3) (4,3.5) (5,2.5) (5,0.5) (2.5,-2) (0,0) (-3,-2) (-3.5,0.5)} node[anchor=west, xshift=0cm, yshift=-0.5cm] {\tiny $\cX_U$};
    \end{scope}

    \draw[very thick, decoration={markings,
    mark=at position 0.5 with \arrow{stealth}},postaction=decorate] (3,4) .. (1, 2) .. (0.8, 0.3) .. (-0.5, -0.6) .. (0, 0);

    \draw[fill opacity=0.35, fill=green!60!blue, text opacity=1] (0,0) circle (0.1cm) node[anchor=west, xshift=-0.1cm, yshift=0.2cm]{\tiny $x^*$};
    
    \draw [very thick, decoration={markings,
    mark=at position 0.5 with \arrow{stealth}},postaction=decorate] (2.3, 4.8) .. (0.2, 4) .. (-0.5, 0.7) .. (0, 0);

    \end{tikzpicture} \label{fig:swa}}}
         \\
    \resizebox{0.7\textwidth}{!}{
         \subfloat[][RWA]{\begin{tikzpicture}[scale=0.45, inner sep=2pt, minimum size=0.4cm,node
    distance=.5cm,
    use Hobby shortcut,
    vec1/.style={
    decorate,
    decoration={
      markings,      %
      mark=at position 0.4 with {\arrow[rotate=5, line width=2pt]{stealth[scale=10]}},
    }
  }
]

    \begin{scope}[xshift=3cm, yshift=4cm, scale=0.3]
    \draw[fill opacity=0.3, fill=blue, text opacity=1]   plot[smooth, tension=.5] coordinates {(-3.5,0.5) (-3,2.5) (-1,3.5) (1.5,3) (4,3.5) (5,2.5) (5,0.5) (2.5,-2) (0,-0.5) (-3,-2) (-3.5,0.5)} node[anchor=east, yshift=0cm, xshift=0cm] {\tiny $\cI$};
    \end{scope}
    
    \begin{scope}[xshift=0cm, yshift=0cm]
    \draw[fill opacity=0.35, fill=green!60!blue, text opacity=1] (0,0) circle (1.1cm) node[anchor=east, yshift=0.4cm, xshift=-0.4cm]{\tiny $\cG$};
    \end{scope}

    \begin{scope}[scale=0.3, xshift=8cm, yshift=2.5cm]
    \draw[fill opacity=0.35, fill=orange, text opacity=1] plot[smooth, tension=.5] coordinates {(-3.5,0.5) (-1,1) (-1,2) (1.5,3) (4,3.5) (5,2.5) (5,0.5) (2.5,-2) (0,0) (-3,-2) (-3.5,0.5)} node[anchor=west, xshift=0cm, yshift=-0.5cm] {\tiny $\cU$};
    \end{scope}

    \draw[very thick, decoration={markings,
    mark=at position 0.5 with \arrow{stealth}},postaction=decorate] (3,4) .. (1, 2) .. (0.8, 0.3) .. (-0.5, -0.6) .. (0.3, 0.5);

    \draw [very thick, decoration={markings,
    mark=at position 0.5 with \arrow{stealth}},postaction=decorate] (2.3, 4.8) .. (0.2, 4) .. (-0.92, 0.7) .. (-2, -1.5);

    \end{tikzpicture} \label{fig:rwa}}
         \subfloat[][RSWA]{\begin{tikzpicture}[scale=0.45, inner sep=2pt, minimum size=0.4cm,node
    distance=.5cm,
    use Hobby shortcut,
    vec1/.style={
    decorate,
    decoration={
      markings,      %
      mark=at position 0.4 with {\arrow[rotate=5, line width=2pt]{stealth[scale=10]}},
    }
  }
]

    \begin{scope}[xshift=3cm, yshift=4cm, scale=0.3]
    \draw[fill opacity=0.3, fill=blue, text opacity=1]   plot[smooth, tension=.5] coordinates {(-3.5,0.5) (-3,2.5) (-1,3.5) (1.5,3) (4,3.5) (5,2.5) (5,0.5) (2.5,-2) (0,-0.5) (-3,-2) (-3.5,0.5)} node[anchor=east, yshift=0cm, xshift=0cm] {\tiny $\cX_I$};
    \end{scope}
    
    \begin{scope}[xshift=0cm, yshift=0cm]
    \draw[fill opacity=0.35, fill=green, text opacity=1] (0,0) circle (1.2cm) node[anchor=west, xshift=-0.7cm, yshift=0.62cm]{\tiny $\mathcal{X}_F$};
    \end{scope}

   \begin{scope}
    \draw[fill, color=white] (0,0) circle (0.7cm) {};
    \draw[pattern={Lines[angle=45, distance=2pt, line width=1pt]}, pattern color=green!60!blue, text opacity=1] (0,0) circle (0.7cm) node[anchor=west, xshift=0.2cm, yshift=-0.8cm]{\tiny $\mathcal{X}_G$};
    \end{scope}
   \draw[] (0.4, -0.4) -- (0.8, -1.5);

    \begin{scope}[scale=0.3, xshift=8cm, yshift=2.5cm]
    \draw[fill opacity=0.35, fill=orange, text opacity=1] plot[smooth, tension=.5] coordinates {(-3.5,0.5) (-1,1) (-1,2) (1.5,3) (4,3.5) (5,2.5) (5,0.5) (2.5,-2) (0,0) (-3,-2) (-3.5,0.5)} node[anchor=west, xshift=0cm, yshift=-0.5cm] {\tiny $\cX_U$};
    \end{scope}

    \draw[very thick, decoration={markings,
    mark=at position 0.75 with \arrow{stealth}},postaction=decorate] (3,4) .. (1, 2) .. (0.8, 0.3) .. (0.4, -1.3) .. (-1, -1.0) .. (-1, 1.0) .. (0,0);
    
    \draw[very thick, decoration={markings,
    mark=at position 0.5 with \arrow{stealth}},postaction=decorate] (2.5,4) .. (0.5, 2) .. (0.2, 0.7) .. (0.1, -0.3) .. (-0.1,-0.1);

    \end{tikzpicture} \label{fig:rswa}}
         \subfloat[][RAR]{\begin{tikzpicture}[scale=0.45, inner sep=2pt, minimum size=0.4cm,node
    distance=.5cm,
    use Hobby shortcut,
    vec1/.style={
    decorate,
    decoration={
      markings,      %
      mark=at position 0.4 with {\arrow[rotate=5, line width=2pt]{stealth[scale=10]}},
    }
  }
]
    \begin{scope}[xshift=3cm, yshift=4cm, scale=0.3]
    \draw[fill opacity=0.3, fill=blue, text opacity=1]   plot[smooth, tension=.5] coordinates {(-3.5,0.5) (-3,2.5) (-1,3.5) (1.5,3) (4,3.5) (5,2.5) (5,0.5) (2.5,-2) (0,-0.5) (-3,-2) (-3.5,0.5)} node[anchor=east, yshift=0cm, xshift=0cm] {\tiny $\cX_I$};
    \end{scope}

    \begin{scope}[xshift=0cm, yshift=0cm]
    \draw[fill opacity=0.35, fill=green!60!blue, text opacity=1] (0,0) circle (0.7cm) node[anchor=east, yshift=-0.8cm, xshift=0.8cm]{\tiny $\mathcal{X}_G$};
    \end{scope}

    \begin{scope}
    \draw[fill opacity=0.35, fill=green, text opacity=1] (0,0) circle (1.2cm) node[anchor=west, xshift=-0.4cm, yshift=0.7cm]{\tiny $\cX_F$};
    \end{scope}

    \draw[] (0.4, -0.4) -- (0.8, -1.5);

    \begin{scope}[scale=0.3, xshift=8cm, yshift=2.5cm]
    \draw[fill opacity=0.35, fill=orange, text opacity=1] plot[smooth, tension=.5] coordinates {(-3.5,0.5) (-1,1) (-1,2) (1.5,3) (4,3.5) (5,2.5) (5,0.5) (2.5,-2) (0,0) (-3,-2) (-3.5,0.5)} node[anchor=west, xshift=0cm, yshift=-0.5cm] {\tiny $\cX_U$};
    \end{scope}

    \draw[very thick, decoration={markings,
    mark=at position 0.75 with \arrow{stealth}},postaction=decorate] (3,4) .. (1.3, 2) .. (0.6, 0.3) .. (0.4, -1) .. (-0.7, -0.7) .. (-0.8, 0.4) .. (-0.3,-0.3) .. (0.3, -0.3);

    \draw[very thick, decoration={markings,
    mark=at position 0.5 with \arrow{stealth}},postaction=decorate] (2,4) .. (-0.1, 2.3) .. (-0.2, 0.8) .. (0.1, 1.3) .. (0.3, 0.8) .. (0.2, 0.2);

    \end{tikzpicture} \label{fig:rar}}
         }
         
     \caption{Pictorial depiction of relevant properties in this work. Here, $\cI$ is the initial set, $\cU$ the unsafe set ($\cS$ is its safe complement), $\cG$ the goal/target set, $\cF$ the final set. (The entire state space is $\cX$.) Here, a dashed background denotes that the corresponding set's existence is implied by the corresponding certificate, but that it is not explicitly defined in the property. \new{Notably, this means that the set cannot be \emph{specified} a-priori when defining the property: e.g., the invariant set $\cG$ may be any size contained within $\cF$. This motivates the construction of the additional certificates (and corresponding properties) ROA and RAR, which instead allow for a-priori set specifications.} }
     \label{fig:properties-example}
\end{figure*}
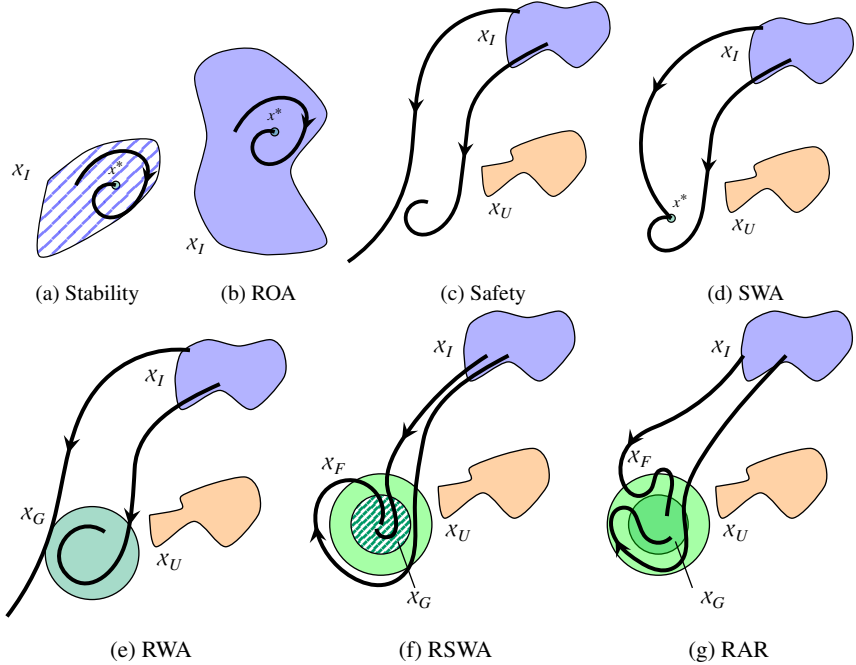

\subsection{Safety}
Safety is another fundamental property we can require from dynamical models. It involves the \emph{avoidance} of some unsafe region:
namely, that no trajectory starting from $\cI$ may enter the unsafe set $\cU$\footnote{We shall later draw connections between the concept of safety and the dual notion of (unconstrained) \emph{reachability}. Please refer to the discussions in  Section   \ref{subsec:classification} and in  \ref{sec:connections}.}; formally 
\begin{equation}
\label{eq:safety-spec}
    \forall \xi(t_0) \in \cI, \forall t  \in \extreal, t \geq t_0, \xi(t) \in {\cU}^\complement.
\end{equation}
For continuous-time models, safety over an unbounded time horizon can be proved via barrier certificates \cite{prajna2004StochasticSafetyVerification, prajna2006BarrierCertificatesNonlinear}. 
\begin{certificate}[Barrier Certificate]
    Consider a dynamical model $f$, a compact unsafe set $\cU$ and compact initial set $\cI$. 
    A function $B: \cX \subset \real^n \to \real, B \in \mathcal{C}^1$, 
    is a Barrier certificate if the following holds: 
    \begin{subequations}
    \label{eq:barr}
    \begin{align}
        &B(x) \leq 0 \ \forall x \in \cI, \\
        &B(x) > 0 \ \forall x \in \cU , \label{eq:barr-b} \\
        &\dot{B}(x) = \langle \nabla B, f(x) \rangle < 0  \ \forall x \in \{x : B(x) = 0 \}. \label{eq:barr-c}
    \end{align}
    \end{subequations}
\end{certificate}
\smallskip

Many characterisations of barrier certificates exist,  to cover different applications or to abide by additional constraints. These include reciprocal \cite{ames2017ControlBarrierFunction}, high-order (zeroing) \cite{tan2022HighorderBarrierFunctions}, or barrier conditions with modifications on the Lie derivative~\eqref{eq:barr-c}~\cite{prajna2004StochasticSafetyVerification}. These certificates are alike in that they certify safety properties. \new{However, the formulation used in this work are known to exist for any system which is safe \cite{ratschan2018ConverseTheoremsSafety}, and hence we consider this sufficient for this work.}

\begin{theorem}[Safety]
\label{thm:barrier}
    Given a model \eqref{eq:dyn-model}, a compact domain $\cX$, a compact initial set $\cI \subset \cX$ and an unsafe set $\cU$, alongside a barrier certificate, then \eqref{eq:safety-spec} holds. \hfill$\blacksquare$
\end{theorem}

\subsection{Stable While Avoid}
It is natural to extend the aforementioned notions of stability and safety towards a combination of both, whereby all relevant trajectories converge towards an equilibrium point, while also avoiding a given unsafe set. Such a property is formally described as follows: 
\begin{multline}
\label{eq:stab-safe-spec}
    \forall \xi(t_0) \in \cI, \exists T \in \extreal, \forall t \in [t_0, T), \xi(t) \in {\cU}^\complement \\
    \wedge \forall \tau \geq T,
    \xi(\tau) \in \{x^*\}.
\end{multline}
Since \eqref{eq:stab-safe-spec} is simply the conjunction of a stability property and a safety property,
certifying this is equivalent to concurrently certifying both stability and safety hold: this task can be thus formally tackled by the following corollary. 
\begin{corollary}[Stable while avoid (SWA)]
\label{cor:stab-safe}
    Given a model \eqref{eq:dyn-model} with unique equilibrium point $x^* \in \cX$,  a compact domain $\cX$, a compact initial set $\cI \subset \cX$ and an unsafe set $\cU$ alongside a ROA certificate $V$ and  barrier certificate $B$, then \eqref{eq:stab-safe-spec} holds. \hfill$\blacksquare$
\end{corollary}

Notice that it is alternatively possible to combine the conditions of \eqref{eq:lyap} and \eqref{eq:barr} into that of a \emph{single} certificate for stability and safety. Such a function is sometimes referred to as a Lyapunov-Barrier certificate \cite{wu2019ControlLyapunovBarrierFunctionbased, romdlony2016stabilization}.  However, in this work we choose to use two separate functions, as this makes synthesis easier and more modular.

\subsection{Reach While Avoid}
Let us now set asymptotic stability aside, and instead study properties over a finite time horizon: namely, we require that trajectories enter a non-singleton set in finite time. These are reachability-like properties, which require trajectories to reach a region 
known as a goal set.
A reach-while-avoid (RWA) property requires a non-singleton goal set to be reached within a finite time horizon $T$, while avoiding an unsafe region; in formal terms, 
\begin{multline}
\label{eq:rwa-spec}
    \forall \xi(t_0) \in \cI, \exists T \in \real, \forall t \in [t_0, T] :
    \\
    \xi(t) \in {\cU}^\complement \wedge \xi(T) \in \cG.
\end{multline}
Next, we introduce an RWA certificate to guarantee that this property holds for a model under consideration.

\begin{certificate}[RWA]\label{def:rwa}
Define an unsafe set $\cU = \cX \setminus \cS$, where $\cS$ is a compact safe set, a compact initial set $\cI \subset \interior(\cS)$, and a compact goal set $\cG \subset \interior(\cS)$ with non-empty interior.
A reach-while-avoid (RWA) certificate \cite{verdier2020FormalSynthesisAnalytic} is a function $V: \real^n \to \real, V \in \mathcal{C}^1$ if there exists $\gamma \in \real_+$,  such that
\begin{subequations}
\begin{align} 
    V(x) \leq 0           \quad & \forall x \in \cI,    \label{eq:rws-cert-a}        \\
    V(x) > 0              \quad & \forall x \in \partial \cS,  \label{eq:rws-cert-b} \\
    \dot{V} (x) \leq - \gamma \quad & \forall x \in \{x \in \cS : V(x) \leq 0 \} \setminus \cG. \label{eq:rws-cert-c}  
\end{align}
\end{subequations}
\end{certificate}

\medskip 

\begin{theorem}[Reach-While-Avoid]\label{thm:rwa} 
    Given a model \eqref{eq:dyn-model} and a RWA Certificate corresponding to the given sets of interest, then \eqref{eq:rwa-spec} holds.  \hfill$\blacksquare$
\end{theorem}

\begin{remark}[Unconstrained Reachability]
Unconstrained reachability can be defined as a special case of RWA, where we set $\cU = \emptyset$ (i.e. $\cS = \cX$). Hence, a certificate can be provided accordingly, as special instance of Certificate \ref{def:rwa}.  \hfill$\blacksquare$
\end{remark}

We emphasise that a RWA certificate does not prove that trajectories will \emph{remain} within the goal set, or that trajectories shall avoid the unsafe set for all (unbounded) time, nor does the specification in \eqref{eq:rwa-spec} indeed encode these requirements. In fact, since the Lie derivative condition \eqref{eq:rws-cert-c} does not hold across the goal set $\cG$, it is possible for trajectories to leave the goal set after entering it, and thereafter possibly enter the unsafe set $\cU$. This alternative, more restrictive scenario is addressed by the next certificate.

\subsection{Reach-and-Stay While Avoid}
A reach-and-stay while avoid (RSWA) property is similar to the RWA property, consisting of RWA with an additional requirement that trajectories will eventually remain within some final set indefinitely. Formally, trajectories should satisfy the property
\begin{multline}
   \label{eq:rswa-spec} 
    \exists \cG \ \new{\subset \interior(\cF)}: \forall \xi(t_0) \in \cI, \exists T \in \real, \forall t \in [t_0,T], \\ \xi(t) \in {\cU}^\complement 
    \wedge \xi(T) \in \cG \\
    \wedge \forall \tau \geq T: \xi(\tau)\in \cF.
\end{multline}
Notably, this property does not require that trajectories reach the final set and remain within it, and may enter and leave the final set as long as they eventually remain within the final set. However, in finite time trajectories must reach some subset of the final set a goal set, after which point they must remain within the final set for all time.
\begin{certificate}[RSWA]
   Define an unsafe set $\cU = \cX \setminus \cS$, where $\cS$ is a compact safe set, then define a compact initial set $\cI \subset \interior(\cS)$, and a compact final set $\cF \subset \interior(\cS)$.
   A RSWA \cite{verdier2020FormalSynthesisAnalytic} certificate is a function $V: \real^n \to \real$, $V \in \mathcal{C}^1$ if there exists $\gamma \in \real_+$ such that
   the following is satisfied:
    \begin{subequations}
    \begin{align}
    V(x) \leq 0           \quad & \forall x \in \cI,    \label{eq:rswa-cert-a}        \\       
    V(x) > 0              \quad & \forall x \in \partial \cS,  \label{eq:rswa-cert-b} \\    
    \dot{V} (x) \leq -\gamma \quad & \forall x \in \{x \in \cS : V(x) \leq 0 \} \setminus \cF, \label{eq:rwsa-cert-c}   \\ 
    V(x) > \beta \quad             & \forall x \in \partial\cF, \label{eq:rswa-cert-d}                                   \\
    \dot{V}(x) \leq -\gamma \quad & \forall x \in \cF \setminus \interior( \{x \in \cS : V(x) \leq \beta \}), \label{eq:rswa-cert-e} 
    \end{align}
    \end{subequations}
    for some constant $\beta \in \real^{-}$.

\end{certificate}

\begin{theorem}[Reach-and-stay while avoid]\label{thm:rswa}
         Given a model \eqref{eq:dyn-model} and a certificate corresponding to the given sets of interest, then \eqref{eq:rswa-spec} holds. \hfill$\blacksquare$
\end{theorem}

The sub-level set of $V$ given by $\beta$ defines an invariant set contained with the final set, and ensures that trajectories reach this set in finite time without entering an unsafe region. Note that the specification described in \eqref{eq:rswa-spec} - and the corresponding certificate - \new{permit trajectories to enter and leave the final set, on the condition that eventually (and within finite time), they enter goal set and do not leave the final set again. Here, the goal set is implicitly defined by the  intersection of $\cF$ and the $\beta$ level set of the certificate, and is an invariant set contained fully within $\cF$. However, the shape and size of this $\cG$ is not specified a-priori, and can only be obtained after synthesis. The next certificate is motivated by a desire to allow $\cG$ and $\cF$ to be specified fully within the property.
}

\subsection{Reach, Avoid and Remain}
The final property we consider is again similar to the previous \emph{Reach and Stay While Avoid} property, but, as with the ROA certificate, we seek to remove the existential quantifier over the goal set from (\ref{eq:rswa-spec}). This means that the Reach Avoid Remain (RAR) property requires that trajectories remain within a final set after reaching a goal set, but for two given goal and final sets. 
We express this formally, as follows: 
\begin{multline}
\label{eq:rar-spec}
    \forall \xi(t_0) \in \cI, 
    \exists T \in \real, \forall t \in [t_0, T]: 
    \\
    \xi(t) \in {\cU}^\complement \wedge \xi(T) \in \cG \wedge \forall \tau \geq T: \xi(\tau) \in \mathcal{X}_F. 
\end{multline}
\begin{certificate}[RAR] 
\label{def:rar}
Define an unsafe set $\cU = \cX \setminus \cS$, where $\cS$ is a compact safe set, a compact initial set $\cI \subset \interior(\cS)$, a compact final $\cF \subset \interior(\cS)$, and  a compact goal set $\cG \subset \interior(\cF)$ with non-empty interior.
Let $V : \real^n \to \real$ be a RWA certificate, and a function $B: \real^n \to \real$, $B \in \mathcal{C}^1$, such that:
\begin{subequations}
\begin{align} \label{eq:rar-barr}
        &B(x) \leq 0 \ \forall x \in \cG, \\
        &B(x) > 0 \ \forall x \in \partial{\cF}  ,\\
        &\dot{B}(x) < 0  \ \forall x \in \{x : B(x) = 0 \}.
\end{align}
\end{subequations}
The pair $(V, B)$ define a Reach-Avoid-Remain certificate.  \hfill$\blacksquare$
\end{certificate}

\medskip

As with the Stable-While-Avoid certificate, we choose to formulate this certificate as a pair of separate functions, rather than collapsing the conditions to a single function. This choice, which of course does not affect the soundness of the approach, renders synthesis practically  easier and more modular. 

\begin{theorem}[Reach-Avoid-Remain]
\label{thm:rar}
    Given a model \eqref{eq:dyn-model} and a certificate pair $V,B$ satisfying the conditions in Certificate \ref{def:rar},  then \eqref{eq:rar-spec} holds. \hfill$\blacksquare$
\end{theorem}

We note that here, the certificate $B$ is similar to a Barrier certificate as defined in \eqref{eq:barr}, though with $\cG$ as the initial set and $\partial{\cF}$ as the unsafe set.  We have restated the function in this context for clarity.

\subsection{Summary and Classification of Properties}
\label{subsec:classification}

So far, we have presented a number of different properties 
that a dynamical model may conform to. 
These properties, and the certificates that sufficiently prove them to hold, can be complex and subtly different. However, we observe the following similarities between them: 
\begin{itemize}
    \item All certificates rely on a set on initial conditions $\cI$. In the case of a Lyapunov function, this set is implicitly defined a-posteriori to the synthesis of the certificate. 
    \item $\cU$ denotes a region trajectories should \emph{avoid} (and thus relate to a safety requirement).
    \item Either $\cG$ or $\{x^*\}$ denote a region which trajectories should enter or \emph{arrive} at. We leverage this notion to encompass both finite-time reachability and asymptotic stability, and note that it can be thought as the dual of the \emph{avoid} category. 
    \item Either $\cF$ or $\{x^*\}$ denote a region which trajectories should eventually \emph{remain} in, for all time, as soon as they have \emph{arrived} in a goal set. This notion is thus related to both (forward) set invariance and asymptotically stable equilibria. 
\end{itemize}

Based on these analogies, we introduce three labels \emph{avoid}, \emph{arrive} and \emph{remain}, and assign them to each certificate, in order to clarify their purpose and differences. We portray this relationships in Fig.~\ref{fig:venn}, 
where the label \emph{arrive} encompasses both stability and finite-time reachability. 

We note that this classification does not distinguish between some certificates, which we clarify here. 
Firstly, as previously mentioned, a Lyapunov function and a ROA certificate differ in that for the latter an initial set is explicitly specified a priori to synthesis. 
Meanwhile, the SWA, RSWA and RAR certificates satisfy all three labels. Stable while Avoid is easy to distinguish, as it handles asymptotic stabilty, whereas the others treat finite time reachability.
Finally, we differentiate the RSWA and RAR certificates as follows. A RSWA certificate proves that there exists a goal set, contained within a given final set, that trajectories will reach and afterwards never leave the final set.
Meanwhile a RAR certificate explicitly defines both the goal set and final set associated with this, allowing for a more elaborate specification.

\begin{figure}[htbp]
    \centering

    \resizebox{\linewidth}{!}{
        
    
    \tikzset{every picture/.style={line width=0.75pt}} %
    \definecolor{dkblue}{rgb}{0.00,0.33,0.68}
    \definecolor{avoid}{RGB}{239, 146, 52}
    \definecolor{arrive}{RGB}{105, 170, 230}
    \definecolor{remain}{RGB}{65, 117, 5}

    \begin{tikzpicture}[x=0.75pt,y=0.75pt,yscale=-1,xscale=1]
    
    \draw  [fill=avoid  ,fill opacity=0.6 ] (298.32,375.8) .. controls (366.2,376.3) and (420.71,447.77) .. (420.07,535.44) .. controls (419.43,623.11) and (363.88,693.77) .. (295.99,693.28) .. controls (228.11,692.78) and (173.6,621.3) .. (174.25,533.64) .. controls (174.89,445.97) and (230.44,375.3) .. (298.32,375.8) -- cycle ;
    \draw  [fill=arrive  ,fill opacity=0.6 ] (78.72,588.66) .. controls (78.17,522.98) and (155.07,469.08) .. (250.48,468.28) .. controls (345.89,467.48) and (423.69,520.08) .. (424.24,585.77) .. controls (424.79,651.45) and (347.89,705.35) .. (252.48,706.15) .. controls (157.07,706.95) and (79.27,654.35) .. (78.72,588.66) -- cycle ;
    \draw  [fill=remain  ,fill opacity=0.4 ] (122.23,586.58) .. controls (122.23,528.32) and (169.47,481.08) .. (227.73,481.08) .. controls (286,481.08) and (333.23,528.32) .. (333.23,586.58) .. controls (333.23,644.85) and (286,692.08) .. (227.73,692.08) .. controls (169.47,692.08) and (122.23,644.85) .. (122.23,586.58) -- cycle ; ;
    \draw  [dash pattern={on 4.5pt off 4.5pt}] (98.48,588.71) .. controls (98.48,542.66) and (135.81,505.33) .. (181.86,505.33) .. controls (227.9,505.33) and (265.23,542.66) .. (265.23,588.71) .. controls (265.23,634.75) and (227.9,672.08) .. (181.86,672.08) .. controls (135.81,672.08) and (98.48,634.75) .. (98.48,588.71) -- cycle ;
    \draw  [dash pattern={on 0.84pt off 2.51pt}] (272.92,593.44) .. controls (272.18,558.03) and (300.28,529.33) .. (335.69,529.33) .. controls (371.1,529.33) and (400.41,558.03) .. (401.15,593.44) .. controls (401.9,628.85) and (373.79,657.55) .. (338.38,657.55) .. controls (302.97,657.55) and (273.66,628.85) .. (272.92,593.44) -- cycle ;

    \draw (83,488.07) node [anchor=north west][inner sep=0.75pt]   [align=left] {};
    \draw (279,344.07) node [anchor=north west][inner sep=0.75pt]   [align=left] {};
    \draw (249,511.07) node [anchor=north west][inner sep=0.75pt]   [align=left] {};
    \draw (270.97,426.66) node [anchor=north west][inner sep=0.75pt]   [align=left] { \textbf{Barrier}};
    \draw (140,621.57) node [anchor=north west][inner sep=0.75pt]   [align=left] { \textbf{Lyapunov}};
    \draw (197.31,561.56) node [anchor=north west][inner sep=0.75pt]   [align=left] {\textbf{Stable}-\\ \textbf{Safe}\\};
    \draw (354.12,589.91) node [anchor=north west][inner sep=0.75pt]   [align=left] {\textbf{RWA}};
    \draw (283,560) node [anchor=north west][inner sep=0.75pt] [align=left] {\textbf{RSWA}};
    \draw (285,609.13) node [anchor=north west][inner sep=0.75pt]   [align=left] {\textbf{RAR}};
    \draw (130,566.32) node [anchor=north west][inner sep=0.75pt]   [align=left] {\textbf{ROA}};
    
    \matrix [below left, xshift=2.3cm, yshift=-0.8cm] at (current bounding box.north west) {
      \node [fill, avoid,label=right:\textbf{Avoid}] {}; \\
      \node [fill, arrive, label=right:\textbf{Arrive}] {}; \\
      \node [fill, remain,label=right:\textbf{Remain}] {}; \\
      \node [label=right:\textbf{Stability}, 
      path picture={
            \draw[dashed] (path picture bounding box.west) -- (path picture bounding box.east);
                    }]{} ;\\
      \node [label=right:\textbf{Reachability}, 
      path picture={
            \draw[dotted] (path picture bounding box.west) -- (path picture bounding box.east);
                    }]{} ;\\
    };
    
    \end{tikzpicture}
    
    }
    \caption{Euler
    Diagram depicting the semantic labels, \emph{arrive, avoid}, and  \emph{remain}, associated with each certificate in this work. By the dashed line, we group properties that exhibit asymptotic stability. By the dotted line, we denote properties that exhibit (finite-time) reachability. 
    }
    \label{fig:venn}
\end{figure}
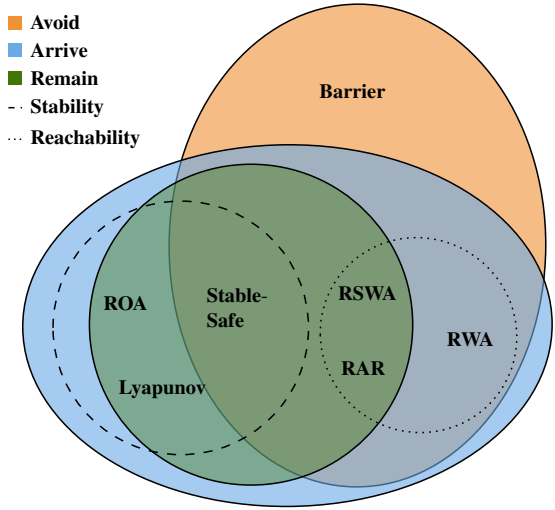

\subsection{Certificates for Control Models}
\label{rmk:control-cert}

Much work in the literature of certificate synthesis refers to, e.g., \emph{control} Lyapunov functions and \emph{control} Barrier certificates. In this work, we consider such terms to concern certificates that refer to a model expressed in terms of both state $x$ and control input $u$, as in  \eqref{eq:con-model}. 
Existing approaches often specify a control set, over which the specifications quantify existentially. In other words, they seek certificates for which there always exists a valid control action that allows for the property to hold. 
A control Lyapunov function therefore proves that there always exists a suitable control input such that the Lyapunov conditions hold, and hence the system is (asymptotically) stabilisable.  
This approach entails that, after a valid control-certificate is synthesised, control actions can be determined, 
e.g. by solving an optimisation program over the input space and the found certificate.

We emphasise that this is \emph{not} the approach taken in this work. Instead, we synthesise a feedback control law for the model described by  \eqref{eq:con-model}, and ``apply'' this state feedback to obtain a closed-loop model of the form of \eqref{eq:dyn-model}, for which we synthesise a certificate. Whilst we synthesise the control law \emph{concurrently} with the certificate, we do not refer to these as ``control certificates'', as in literature. Hence, we verify properties for control models using both a controller and a certificate, and in particular do not delegate the controller synthesis a-posteriori. 

\new{We note that in general neither approach has a clear benefit over the other, and more variants exist in literature. However, verifying a control-Lyapunov certificate based on an exists/for-all query would, in general, scale poorly with SMT solvers. As such, we prefer to break the problem into two tasks of \emph{guessing} a candidate controller and \emph{checking} it with a  certificate only dependent on the closed-loop dynamics.}

\section{Synthesis of Certificates and Controllers}
\label{sec:synthesis}

Following the introduction of specifications that are salient for verification purposes, and certificates whose existence prove that a given model satisfies these conditions, we describe next a unified, efficient, and sound algorithm for the synthesis of these certificates, for both dynamical and controlled models. 
The objective of the synthesis task is indeed twofold: we seek a feedback controller $k(x)$ that ensures a control model satisfies a desired specification, and concurrently we synthesise a certificate $C(x)$ that serves as a proof that the specification holds. As such, allow us to use $C(x)$ to refer to each function in a possible pair separately.
Call with $\theta_u$ the parameters of the state feedback law
and $\theta_c$ the parameters of the certificate function. The synthesis task therefore amounts to finding values for the parameters $\theta_u$ and $\theta_c$ such that the certificate conditions hold, i.e.,
\begin{equation}
    \exists \ \theta_u, \theta_c \ \forall x \in \cX :  \ \phi(C(x)),
\end{equation}
where $\phi(C(x))$ denotes the conditions, related to a given specification (as formalised in the previous section), that the certificate ought to satisfy. 
Our procedure is based on counter-example guided inductive synthesis (CEGIS) \cite{solar-lezama2006CombinatorialSketchingFinite}, an established approach to formally solving such exists-forall queries. CEGIS consists of two opposing components: a \emph{leaner} and a \emph{verifier} (cf. Figure \ref{fig:cegis}), which we detail in turn.

\subsection{Learner}

The learner fulfils the twofold task of designing both a stabilising control law (when required) and the desired certificate, as exemplified in Fig.~\ref{fig:learner-scheme}.

In this work, we seek \emph{feedback} control laws and employ neural networks as a general template for these functions:  neural architectures allow for a wide variety of control laws, as determined by the choice of their activation functions and of the number of neurons.  

Our candidate controller is thus simply the output of a neural network, contributing to the closed-loop dynamics, and which is then employed within the certificate synthesis procedure.  

For the certificate synthesis, we also employ a neural network as a template for $C(x)$: this allows not only for polynomial templates akin to classical (e.g., SOS-based) certificates, but also for highly non-linear functions, leveraging the expressive power of neural networks. 
One of the crucial components of a successful training within the learner is the definition of a tailored loss function.

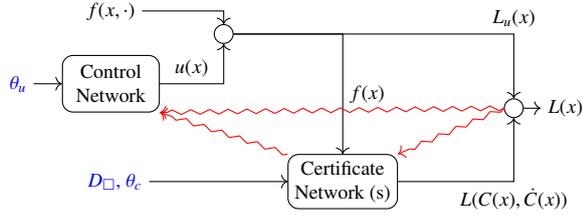
\begin{figure}[t]
    \centering
    \usetikzlibrary{arrows}
\usetikzlibrary{calc}
\usetikzlibrary{decorations.pathmorphing}

\begin{tikzpicture}[scale=0.75, transform shape]

    \node[draw, rectangle, rounded corners, minimum width=1.7cm, minimum height=1cm, align={center}] (ctrl) at (0, 0) {Control \\ Network};

    \node[left = 0.5cm of ctrl, align={center}, color=blue] (ctrl-input) {$\theta_u$};

    \node[align={center}, above = 0.5cm of ctrl] (fx) {$f_u(x, \cdot)$};

    \node[draw, circle, above right = 0.25cm and 1cm of ctrl] (sum) {};

    \node[draw, rectangle, rounded corners, below right = 2cm and 1cm of sum, align={center}] (lyap-net) {Certificate \\ Network (s)};

    \node[left of = lyap-net, node distance=4cm, color=blue] (inputs) {$D_{\square}$, $\theta_c$};

    \node[right of = lyap-net, node distance=3cm] (loss) {};

    \node[draw, circle, above = 1cm of loss] (loss-sum) {};

    \node[right = 0.3cm of loss-sum] (end) {$L(x)$};

    \node[above right = -0.1cm and 1cm of lyap-net] (above2) {};

    \draw[->] (ctrl-input) -- (ctrl);
    \draw[->] (fx) -| (sum);
    \draw[->] (ctrl) -|   node[near start,above]{$u(x)$}  (sum);

    \draw[->] (sum) -| node[near end, right]{$f(x)$} (lyap-net);

    \draw[->] (sum) -| node[above]{$L_u(x)$} (loss-sum);

    \draw[->] (lyap-net) -| node[below]{$L(C(x), \dot{C}(x))$} (loss-sum);

    \draw[->] (inputs) -- (lyap-net);

    \draw[->] (loss-sum) -- (end);

    \draw [->, red,
            line join=round,
            decorate, decoration={
                zigzag,
                segment length=7,
                amplitude=.9,post=lineto,
                post length=2pt}]
                (loss-sum.west) -- (lyap-net.north east);

    \draw [->, red,
            line join=round,
            decorate, decoration={
                zigzag,
                segment length=7,
                amplitude=.9,post=lineto,
                post length=2pt}]
                (lyap-net.north west) -- (ctrl.south east);

        \draw [->, red,
            line join=round,
            decorate, decoration={
                zigzag,
                segment length=7,
                amplitude=.9,post=lineto,
                post length=2pt}]
                (loss-sum.west) -- (ctrl.south east);

  \end{tikzpicture}
    \caption{Loss calculation block diagram. Blue indicates the inputs required for the loss calculation: namely the sampled data sets $D_\square$, the control and certificate network parameters $\theta_u$ and $\theta_c$, respectively. 
    Red arrows represent the back-propagation steps, updating the networks' parameters during training. $f_u(x, \cdot)$ corresponds to the model in \eqref{eq:con-model}, while $f(x)$ corresponds to the model in \eqref{eq:dyn-model}. 
    }
    \label{fig:learner-scheme}
\end{figure}

\subsubsection{Certificate Loss}
We note that all of the certificate conditions described in the previous section can be expressed in terms of inequalities, either as
\begin{equation}
\begin{split}
    C(x) \bowtie c \ \forall x \in \mathcal{X}_{\square},
\end{split}
\text{  or as }
\begin{split}
    \dot{C}(x) \bowtie c \ \forall x \in \mathcal{X}_{\square}, 
\end{split}
\end{equation}
where $\bowtie \coloneqq \{<, >, \geq, \leq \}$, $C(x)$ represents the certificate, $\mathcal{X}_{\square}$ represents the relevant set (e.g., $\cI$ or $\cX$)
and $c \in \real$.
We can then construct a loss function based on this observation, as follows. Consider a monotonically increasing function $m(\cdot)$ and 
suppose we have a finite set of sampled data points over the set $\mathcal{X}_{\square}$, which we denote as $D_\square$. A general loss function for any of the discussed certificate conditions is thus 
\begin{equation}
    \sum_{d \in D_\square} m(p \cdot C(d)),
\end{equation}
where $p = 1$ if $\bowtie \, = \{\leq, < \}$ and $p = -1$ if $\bowtie \, = \{\geq, > \}$. Note that this function penalises points in $D_\square$ where the required condition is not satisfied. 
Suitable choices for function $m$ are leaky-ReLU, which is piecewise linear, and softplus, which is smooth. 
Since several of the sets $\mathcal{X}_{\square}$ are boundaries of sets, or represent level sets, in practice we consider a small band around them, in order to encompass a sufficient number of data points in $D_\square$. 

\begin{example}
\emph{Example.}
Let us consider a barrier certificate $B(x)$ for safety verification (see  \eqref{eq:barr}). We create separate sets of finite samples for each set defined by the property (initial, unsafe, state-space), and the resulting loss function is given by
\begin{multline}
L = 
\frac{1}{N_I} \sum_{d \in D_I} m(B(d))
+ 
\frac{1}{N_U} \sum_{d \in D_U} m(-B(d))
\\
+ 
\frac{1}{N} \sum_{d \in Z_{B}} m(\dot{B}(d)).
\label{eq:L0-term}
\end{multline}
Here $Z_{B}(d) = \{ d \in D : |{B}(d)| \leq \epsilon \}$, where $D$ is the set of $N$ samples over the whole state space, $D_I$ is the set of $N_I$ samples over the initial set and $D_U$ is the set $N_U$ of samples of the unsafe set. \hfill $\blacksquare$

\end{example}

\subsubsection{Controller Loss}\label{sec:cont-loss}
Let us now discuss controlled models. 
Note that, as described in Figure \ref{fig:learner-scheme}, the parameters of the neural network controller appear in the loss function for the certificate, as described previously. Specifically, they manifest themselves in the terms corresponding to conditions on the Lie derivative of the certificate,  since these in turn depend on the control-dependent dynamics $f$. This should encourage the learner to seek a controller that enables a valid certificate. 
However, in practice we find that this is insufficient for robust synthesis in the case of \emph{arrive} requirements. While our certificates and properties make no assumption on the location of the goal set or of the equilibrium, 
oftentimes an equilibrium point (without loss of generality, the origin) lies within the goal set. 
If this is the case, 
trajectories converging to the origin exhibiting a negative Lie derivative are in turn attracted to the goal set.
Since the Lie derivative consists of two components, $f(x)$ and $\nabla C(x)$, it can be helpful to separate these elements within the loss function to encourage better learning. 
We thus propose the additional term for the loss function: 
\begin{equation}
\label{eq:ctrl-loss}
    L_u = \frac{1}{N} \sum_{d \in D} \frac{\langle d, f(d) \rangle}{||d|| \cdot ||f(d)||},
\end{equation} 
where  $\langle \cdot, \cdot \rangle$ is the inner product of its inputs and $||\cdot||$ is the 2-norm of its input, and $D$ is the set of $N$ samples over the state space. This loss is known as the \emph{cosine similarity} of the two vectors $d$ and $f(d)$, and rewards vectors for pointing in opposite directions. 
We interpret this loss function as separating $f(x)$ from $\nabla C(x)$ when calculating the Lie derivative, and instead replacing the corresponding certificate with a default positive definite function (that of Euclidean distance). This encourages the loss function to specifically focus on the parameters in the feedback law to learn a desirable $f(x)$. As an alternative interpretation, since any vector $d$ is fixed and points away from the origin, this encourages a controller which guides the dynamics to point towards to origin, and hence eventually converge towards it. 
We demonstrate the efficacy of this loss component in Section \ref{sec:res-cont}.

\subsection{Verifier}
\label{sec:cegis-subsec:verifier}
We now discuss the dual component of the CEGIS architecture, as portrayed is Figure \ref{fig:cegis}.  The verifier's role is assumed by an SMT solver and its operation is described as follows: let $\phi$ denote the conditions required for a given certificate. We seek a point $x \in \cX$  that violates any of the constraints $\phi$ associated to the certificate. To this end, we express the {\it negation} of such requirements $\phi$, and formulate a nonlinear constrained problem over real numbers. Formally, we ask an SMT solver to find a witness to 
\begin{equation}
    \label{smt-spec}
    \exists x \in \cX: \neg \phi(C(x)),
\end{equation} 
where any such witness $x$ would be considered a counterexample to the validity of the certificate $\phi$.

\begin{example*}[Example Cont'd]
 Consider the negation of barrier certificate conditions $\phi(C(x))$ as in  \eqref{eq:barr}, namely 
\begin{multline}
(x \in \cI \wedge B(x) > 0) \vee
(x \in \cU \wedge B(x) \leq 0) \vee
\\
(B(x) = 0 \wedge \dot{B}(x) \geq 0).
\label{eq:barrier-smt}
\end{multline}
The verifier searches for solutions $x$ of the constraints in  \eqref{eq:barrier-smt}.   
This in general requires manipulating non-convex 
functions and is therefore handled by an SMT solver. 
Whenever such $x$ is found, it represents a witness that the candidate $B(x)$ is \emph{not} a valid certificate function. 
\hfill $\blacksquare$

\end{example*}

\smallskip

The correctness of our algorithm hinges upon the soundness of the verification engine: we use two solvers, Z3 \cite{demoura2008Z3EfficientSMT} and dReal \cite{gao2013dRealSMTSolver}, both of which are sound over nonlinear real arithmetic.  
Z3 is restricted to polynomial reasoning, whereas dReal can handle non-polynomial expressions, for instance containing trigonometric or exponential terms, thus allowing for more complex models and certificates (via their activation functions).   
We note that dReal is a $\delta$-complete SMT-solver: while
this guarantees dReal will always find a counterexample if one exists, it may return also spurious counterexamples within a $\delta$-perturbation of the original formula \cite{gao2013dRealSMTSolver}.
This implies that our procedure may not terminate, even when the candidate is valid. 
However quite importantly it does not compromise the correctness of certificates we verify successfully.

\subsection{Enhanced Communication amongst Components}

Our CEGIS approach builds on that of the software tool \textsf{Fossil}~\cite{abate2021FOSSILSoftwareTool}.
As part of its CEGIS loop, \textsf{Fossil} adds two elements to enhance the communication between components. %

\textit{Translator.}
The translator is tasked with the conversion of the neural networks into a symbolic candidate $C(x)$
and the corresponding $\dot{C}(x)$, ready to be processed by the verifier (see Fig.~\ref{fig:cegis}).
The efficiency of SMT solvers depends also upon the numerical expressions of the formulae to be verified. Oftentimes the training returns numerically ill-conditioned expressions, e.g.  unreasonably small coefficients (e.g., in the order of $10^{-8}$);
these candidates might slow the verification step, and thus the whole procedure.
The translator thus rounds the coefficients of the candidate function to a specified precision, in order to help human interpretation and the verification process. \new{Notably, this rounding step is performed \emph{before} the SMT-based verification step (cf. Fig. \ref{fig:cegis}, meaning the correctness of the synthesised certificates and controllers is not compromised.}

\textit{Consolidator.}
Generating counterexamples is, in general, an expensive procedure, and the verification engine returns a single counterexample (\texttt{cex} in Fig.~\ref{fig:cegis}); 
e.g. 
an instance satisfying  \eqref{eq:barrier-smt}. 
Naturally, an isolated sample does not provide enough information for the learner to improve the candidate certificate. 
To overcome this issue,  
we randomly generate a \textit{cloud} of points around the \texttt{cex} point, 
since samples around a counterexample are also likely to invalidate the certificate conditions.  
Secondly, starting from \texttt{cex}, we compute the gradient of $C$ (or of $\dot{C}$) thanks to an automatic differentiation feature, and follow the direction that maximises the violation of the certificate constraints. 
The consolidator then aggregates 
the original counterexample and the newly generated points (denoted \texttt{cex$^+$} in Fig.~\ref{fig:cegis}) with the sample set $S$ for a new synthesis round by the learner.

\begin{figure}[htbp]
  \centering
  \begin{tikzpicture}[scale=0.9, transform shape]
  
  \node[draw, dashed, rectangle, rounded corners, align={center}, minimum width=8cm, minimum height = 4.5cm] (cegis) at (0, 0) {}; %
    \node[rectangle,draw, rounded corners, minimum width=1.7cm, minimum height=1cm] at (-2, 1) (lea)  {Learner};
    \node[rectangle, draw, rounded corners, minimum width=1.7cm, minimum height=1cm, below=1.5cm of lea] (simpli)  {Translator};
    \node[rectangle,draw, rounded corners, minimum width=1.7cm, minimum height=1cm] at (2, 1) (traj)  {Consolidator};
    \node[rectangle,draw, rounded corners, minimum width=1.7cm, minimum height=1cm, below=1.5cm of traj] (ver)  {Verifier};

        \node[] at (0, 2) {CEGIS};
		\node[right=0.5cm of ver] (out) {};
		\node[left=1.5cm of lea] (ool) {};
        
    \draw[->] (lea) edge node[left, align=center]{$C(x)$ \\ (Neural Net)} (simpli);
    \draw[->] (simpli) to node[above, align=center] {$C(x)$ \\ (Symbolic)} (ver);
    \draw[->] (ver) edge node[right]{\texttt{cex}} (traj);
    \draw[->]
    (traj) edge node [above]{$D_\square \cup \texttt{cex}^+$} (lea);
    \draw[->] (ver) to node[above,xshift=.3cm,yshift=.1cm] {\tt valid} (out);
    \draw[->] (lea) edge node[above, xshift=0.2cm, align={center}]{\texttt{out of} \\ \texttt{loops}} (ool);
    
  \end{tikzpicture}
  \caption{Enhanced CEGIS architecture within \textsf{Fossil}.}
  \label{fig:cegis}
\end{figure}
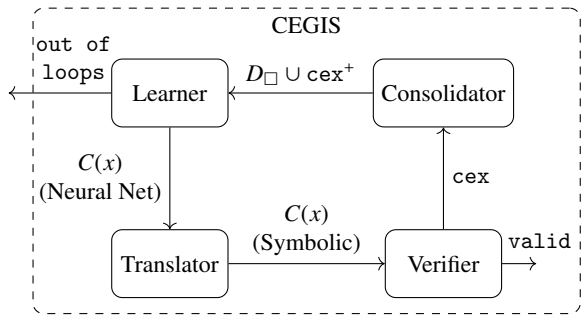 

\subsection{Comments on Specific Certificates}
We add a few remarks next, elaborating on practical synthesis details that are unique to specific  certificates.

\subsubsection{Lyapunov functions}
We assume without loss of generality that the given model has an equilibrium at the origin, \new{and that $k(0) = 0$ for the control law}. We aid synthesis by enforcing the positivity condition by construction: in the case of positive-definite polynomial functions this is done by fixing the output layer's weights $W_{k+1}$ to be all positive.

\subsubsection{ROA}
For verification, we estimate the smallest level set that contains $\cI$ using sample points. Then, we verify that $\cI$ is contained wholly within this level set and that the Lyapunov conditions hold over the entire level set. Due to dReal's inner workings, when verifying either Lyapunov or ROA  certificates, we remove a very small region around the origin from the verification domain. This is common practice \cite{chang2020NeuralLyapunovControl} across similar works, and does not affect sub-level sets outside this region. In fact, this notion is studied as $\epsilon$-stability \cite{gao2019NumericallyrobustInductiveProof}. This caveat  does not apply when using Z3.

\subsubsection{RSWA}
The choice of $\gamma$ for the RSWA and RWA certificates is arbitrary and may be fixed prior to synthesis. Meanwhile, the conditions described in (\ref{eq:rswa-cert-d}, \ref{eq:rswa-cert-e}) depend on the existence of the parameter $\beta$. 
We must prove a value for this parameter exists such that the relevant conditions hold in order to prove the certificate is valid. 
After successfully verifying the conditions which do not depend on $\beta$, we perform a line search for $\beta$ to find a suitable value \cite{verdier2020FormalSynthesisAnalytic}. 
If we do not find a valid $\beta$, we return to synthesis and keep training.

\section{Computational Experiments and Benchmarks}
\label{sec:experiments}

We have implemented the proposed framework based on the computational library of \textsf{Fossil}. 
We have tested our approach across a large number of case studies, and benchmarked it against the first release of \textsf{Fossil}, showing improved results.

\begin{table*}[htbp] 
    \centering
    \resizebox{\linewidth}{!}{
    \begin{tabular}{lrrllllllr}
    \toprule
       & $N_s$ & $N_u$ & Property & Neurons        & Activations                                                                      & \multicolumn{3}{c}{$T$ (s)} &                         \\ \cmidrule(lr){7-9}
       &       &       &             &                &                                                                                  & $\min$                      & $\mu$           & $\max$         &  S (\%) \\
    \midrule
    1  & 2     & 0     & Stability   & [6]            & [$\varphi_{2}$]                                                                  & 0.01 ($\approx$ 0.00)       & 0.16  (0.15)    & 1.50   (1.48)  & 100 \\
    2  & 3     & 0     & Stability   & [8]            & [$\varphi_{2}$]                                                                  & 0.28 ($\approx$ 0.00)       & 2.22  (0.45)    & 12.57  (3.31)  & 100 \\
    3  & 2     & 2     & Stability   & [4]            & [$\varphi_{2}$]                                                                  & 0.07  (0.01)                & 0.19  (0.02)    & 0.47   (0.04)  & 100 \\
    4  & 2     & 2     & Stability   & [5]            & [$\varphi_{2}$]                                                                  & 0.09  (0.01)                & 0.26  (0.02)    & 0.54   (0.03)  & 100 \\ \midrule
    5  & 2     & 0     & ROA       & [5]            & [$\sigma_{\mathrm{t}^2}$]                                                        & 0.71  (0.02)                & 1.17  (0.02)  & 2.59   (0.03)  & 50     \\
    6  & 3     & 3     & ROA         & [8]            & [$\varphi_{2}$]                                                                  & 1.24  (0.02)                & 39.08 (0.03)    & 287.89 (0.04)  & 100 \\ \midrule
    7  & 2     & 0     & Safety      & [15]           & [$\sigma_{\mathrm{t}}$]                                                          & 0.44  (0.35)                & 3.36  (2.90)    & 7.61   (7.11)  & 100 \\
    9  & 8     & 0     & Safety      & [10]           & [$\varphi_{1}$]                                                                  & 12.63 (7.71)                & 51.97 (32.75)   & 70.59  (44.66) & 70  \\
    10 & 3     & 1     & Safety      & [15]           & [$\sigma_{\mathrm{t}}$]                                                          & 1.57  (0.19)                & 11.87 (2.50)    & 51.08  (7.52)  & 90  \\ \midrule
    11 & 3     & 0     & SWA         & [6], [5]       & [$\varphi_{2}$], [$\sigma_{\mathrm{t}}$]                                         & 0.19  (0.05)                & 2.46  (0.100)   & 12.10  (0.20)  & 90  \\
    12 & 2     & 0     & SWA         & [5], [5, 5]    & [$\varphi_{2}$], [$\sigma_{\mathrm{sig}}$,$\varphi_{2}$]                         & 0.13  (0.06)                & 0.27  (0.14)    & 0.39   (0.20)  & 100 \\
    13 & 2     & 1     & SWA         & [8], [5]       & [$\varphi_{2}$], [$\varphi_{2}$]                                                 & 0.06  (0.03)                & 0.20  (0.10)    & 0.58   (0.24)  & 90  \\
    14 & 3     & 1     & SWA         & [10], [8]      & [$\varphi_{2}$], [$\sigma_{\mathrm{t}}$]                                         & 4.06  (0.87)                & 19.81 (2.73)    & 103.49 (7.23)  & 90  \\ \midrule
    15 & 2     & 0     & RWA         & [4]            & [$\varphi_{2}$]                                                                  & 0.14  (0.09)                & 1.81  (1.75)    & 4.70   (4.63)  & 100 \\
    16 & 3     & 0     & RWA         & [16]           & [$\varphi_{2}$]                                                                  & 1.36  (0.09)                & 14.10 (0.14)    & 72.97  (0.20)  & 90  \\
    17 & 2     & 1     & RWA         & [4, 4]         & [$\sigma_{\mathrm{sig}}$,$\varphi_{2}$]                                          & 0.59  (0.27)                & 6.82  (3.32)    & 20.07  (11.46) & 100 \\
    18 & 3     & 1     & RWA         & [5]            & [$\varphi_{2}$]                                                                  & 0.46  (0.11)                & 16.06 (5.81)    & 72.47  (44.64) & 80  \\
    19 & 2     & 2     & RWA         & [5]            & [$\sigma_{\mathrm{sig}}$]                                                        & 0.69  (0.40)                & 1.38  (0.94)    & 2.14   (1.90)  & 100 \\ \midrule
    20 & 2     & 0     & RSWA        & [4]            & [$\varphi_{2}$]                                                                  & 0.19  (0.03)                & 1.29  (1.04)    & 3.79   (3.37)  & 100 \\
    21 & 3     & 0     & RSWA        & [16]           & [$\varphi_{2}$]                                                                  & 4.81  (0.13)                & 27.14 (0.19)    & 80.95  (0.25)  & 100 \\
    22 & 2     & 0     & RSWA        & [5, 5]         & [$\sigma_{\mathrm{sig}}$,$\varphi_{2}$]                                          & 1.52  (0.06)                & 4.45  (0.19)    & 10.97  (0.35)  & 100 \\
    23 & 2     & 1     & RSWA        & [8]            & [$\varphi_{2}$]                                                                  & 0.21  (0.05)                & 0.67  (0.25)    & 1.19   (0.91)  & 100 \\
    24 & 2     & 2     & RSWA        & [5, 5]         & [$\sigma_{\mathrm{sig}}$,$\varphi_{2}$]                                          & 0.98  (0.16)                & 1.23  (0.28)    & 1.61   (0.46)  & 100 \\ \midrule
    25 & 2     & 0     & RAR         & [6], [6]       & [$\sigma_{\mathrm{soft}}$], [$\varphi_{2}$]                                      & 6.65  (1.08)                & 24.74 (6.46)    & 77.80  (15.06) & 100 \\
    26 & 2     & 2     & RAR         & [6, 6], [6, 6] & [$\sigma_{\mathrm{sig}}$,$\varphi_{2}$], [$\sigma_{\mathrm{sig}}$,$\varphi_{2}$] & 5.13  (1.34)                & 26.99 (9.90)    & 101.23 (60.14) & 100 \\
    \bottomrule
\end{tabular}

    }
    \caption{Results of synthesising certificates for all properties presented in this work. The first column indexes the benchmarks. $N_s$: Number of states, $N_u$: number of control inputs. We show the \emph{Property} being verified and network structure (\emph{Neurons} and \emph{Activations}). For certificates of two functions, comma-separated lists shows the different structures. Finally, we report success rate ($S$) and the minimum, mean ($\mu$) and maximum computation time $T$ over successful runs, in seconds. In brackets we show the time spent during the \emph{learning} phase.}
    \label{tab:main}
\end{table*}

\subsection{Main results}\label{sec:res-main}

Our new prototype tool, which we refer to as \textsf{Fossil} 2.0\footnote{The features and interface of \textsf{Fossil} 2.0 are detailed in \cite{fossil2}. We note for clarity that \textsf{Fossil} 2.0 is able to prove some properties for discrete-time models, but that this is not relevant to the topic of study in this work: continuous time dynamical models.}, is able to verify all properties described in Section \ref{sec:properties} for continuous-time models, both autonomous and controlled. 
We showcase the efficacy of our framework and corresponding tool across 26 benchmarks, borrowed from existing literature on certificate synthesis \cite{sankaranarayanan2013LyapunovFunctionSynthesis, verdier2020FormalControllerSynthesis, abate2021FOSSILSoftwareTool, vannelli1985MaximalLyapunovFunctions}. Note that, in some cases, we have modified these benchmarks to further challenge our approach, for instance by using disjoint, non-convex sets in the specifications. We consider a key strength of our approach to be its flexibility - we are able to perform well on straightforward and challenging benchmarks using certificates that represent both polynomials and more complex non-polynomial functions (as determined by the activation function of the neural network). We reflect this in our selection of benchmarks, including dynamics that are relatively simple and dynamics that involve transcendental and trigonometric functions. Due to the large number of benchmarks, details on the dynamics and sets can be found in an extended version of this paper \cite{edwards2023GeneralVerificationFramework}, and in the corresponding code-base, \href{https://github.com/oxford-oxcav/fossil}{https://github.com/oxford-oxcav/fossil}, where additional benchmarks can be also found.

The results are reported in Table~\ref{tab:main}, where for each benchmark we outline the number of variables $N_s$ and of control input $N_u$, the property to be verified (cf. acronyms introduced earlier), the number of neurons in each hidden layer and the corresponding activation functions for these layers. 
The number of neurons is denoted as a list, e.g. $[n_1, n_2]$ indicates that the first and second hidden layers are composed of $n_1$ and $n_2$ neurons, respectively. The activation functions for these hidden layers are denoted similarly. As mentioned, a strength of our methodology is its flexibility in terms of the form that certificates may take: we are able to synthesise polynomial certificates as well as non-polynomial certificates that represent more ``neural-typical'' functions -- this is illustrated in the ``Activations'' column of Table~\ref{tab:main}. By $\varphi_j$ we denote that the layer represents a polynomial function of order $j$; $\sigma_\textrm{sig}$ represents the sigmoid function, $\sigma_\textrm{t}$ represents the hyperbolic tangent function, $\sigma_{\textrm{t}^2}$ is the square of $\sigma_\textrm{t}$ and $\sigma_\textrm{soft}$ is the softplus function.

For almost all benchmarks, we use a linear control function. Our approach can handle more general nonlinear templates, but we emphasise that, as we solve a verification problem, rather than a control problem, we only seek a feedback law such that the property is satisfied by the closed-loop dynamics, and thus offer no guarantee on the optimality of this controller. Still, we use a nonlinear controller employing $\sigma_\textrm{t}$ functions for the benchmark number 10 of Table \ref{tab:main}.

We measure the robustness performance by running each experiment 10 times, where we initialise the network with different weights and a new dataset across separate random seeds. Our procedure is not guaranteed to terminate, so after a maximum number of CEGIS loops we stop it and consider the overall run a failure. In general, we allow 25 CEGIS loops; for the SWA and RAR certificates we allow 100 CEGIS loops as these certificates are composed of two functions.

We consider two metrics to assess the quality of our framework when attempting to synthesise a certificate and controller: how often it returns a successful result, the success rate $S$, and how long the algorithm takes for these successful runs, the time $T$. 
Table~\ref{tab:main} thus reports $S$, along with the average, minimum and maximum time for the procedure to terminate, under the $T$ column, denoted $\mu$, $\min$ and $\max$, respectively. In brackets, we also denote the amount of time spent during the learning phase of our procedure, with approximately all remaining time spent during the verification phase.

The success rate is consistently close to 100\%, with the minimum standing at 50\% for a single benchmark, highlighting the robustness of our approach over the presented broad range of complex properties. \new{The benchmark with a success rate of 50\% is \emph{Benchmark 5}, and is shown in Fig. \ref{fig:roa}. The difficulty of this benchmark is due to the nature of the dynamics, which admit only a non-polynomial global Lyapunov function (studied in \cite{vannelli1985MaximalLyapunovFunctions}). Further, we seek to test our approach by seeking a region of attraction for this model with a \emph{disjoint}, and thus non-convex, initial set. It is clear from the figure that the region of attraction cannot become too wide, or else it might fail over points of instability. We include a simple Lyapunov function for the model to demonstrate the benefit of using a ROA certificate. }

\begin{figure*}
    \centering
    
         \subfloat[\new Surface of a Lyapunov function (Benchmark 4) with the proven region of attraction shown in dashed line. This ellipsoidal shape is a typical example which motivates the use of the ROA certificate to enforce a specific region as stable.]{\includegraphics[width=0.3\linewidth]{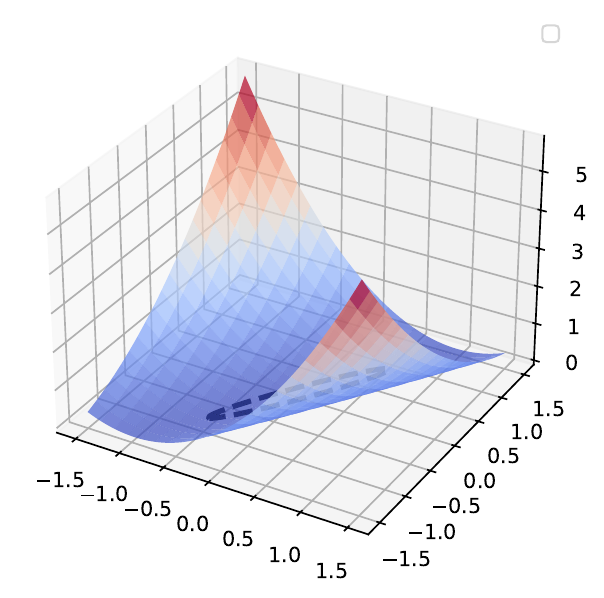} \label{plot:lyap}}   
         \hfill
         \subfloat[\new ROA certificate (Benchmark 5) showing the largest verified sub-level set and corresponding region of attraction, with the level set (smaller dashed line) of a synthesised Lyapunov function shown for comparison.]
         {\includegraphics[width=0.3\linewidth]{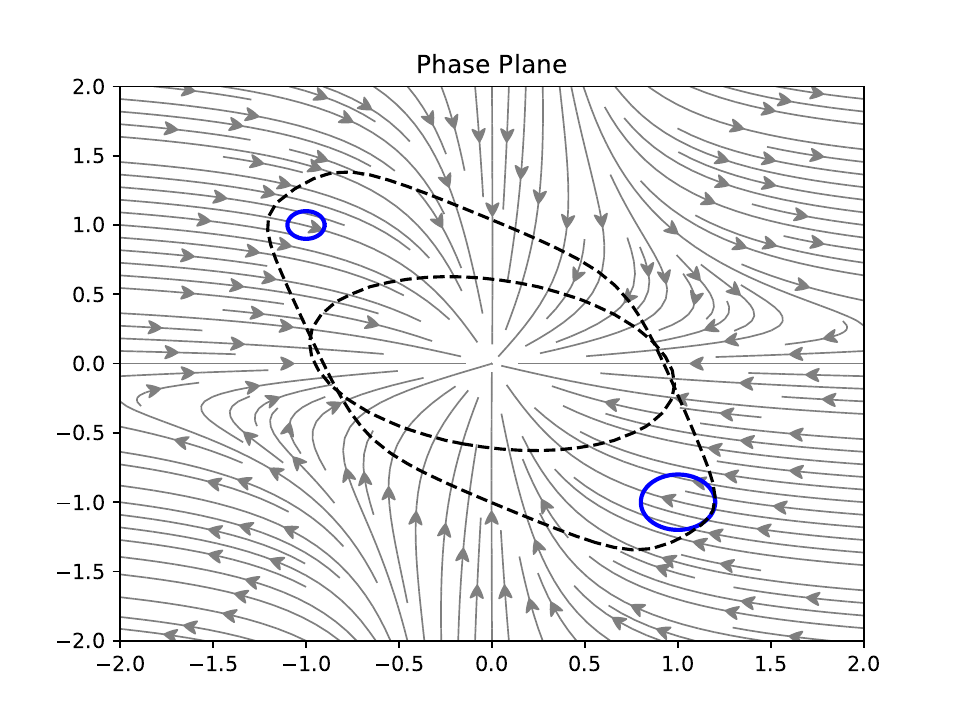} \label{plot:roa}}
         \hfill
         \subfloat[\new Barrier certificate (Benchmark 7) showing the initial (blue) and unsafe (red) overlaid on a phase portrait, and the zero contour of the barrier certificate (dashed black)]{\includegraphics[width=0.3\linewidth]{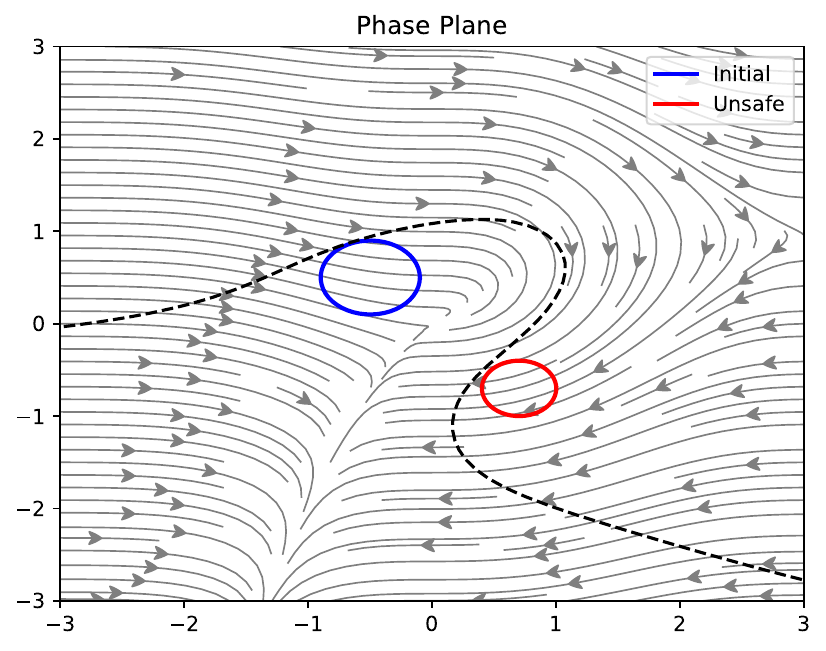} \label{plot:barr2}}
         \\
         
         \subfloat[\new Surface plot of a RWA certificate (Benchmark 17) with zero level set as dashed line. The light blue set depicts a nonconvex safe set (note the circular region which is unsafe, inside the safe rectangular set). The safe set is the complement of the unsafe set; we choose to show the safe set here as it is more clear in the context.]{\includegraphics[width=0.3\linewidth]{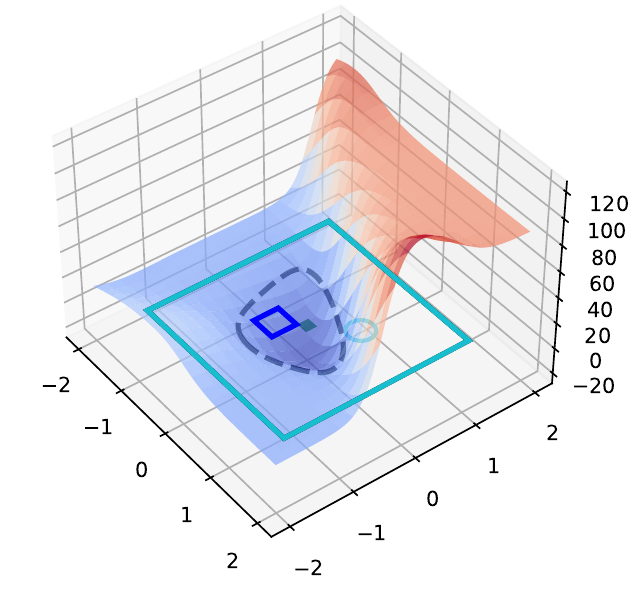} \label{plot:rwa}}
         \hfill
         \subfloat[\new Surface plot of a SWA certificate (Benchmark 12) showing the constituent ROA (left) and Barrier (right) components of the certificate. The region of attraction and zero contour are shown of the ROA and Barrier components respectively in dashed black line.]{\includegraphics[width=0.3\linewidth]{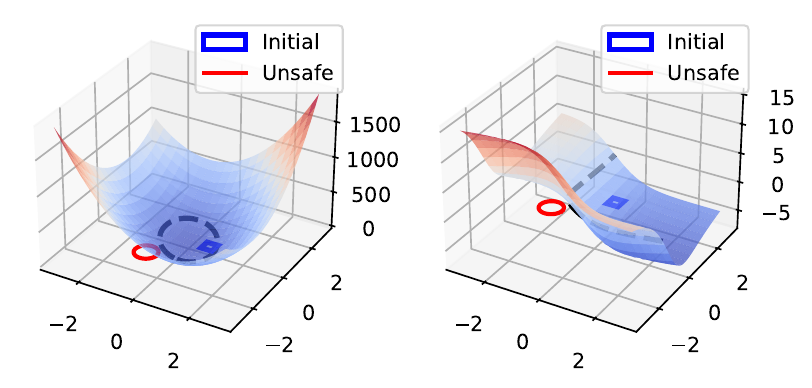} \label{plot:SWA}}
         \hfill
         \subfloat[\new RAR certificate (Benchmark 26) with zero level sets for both the RWA function (outer dashed line) and barrier component (inner dashed line). See the main caption figure for a description of the coloured sets.]{\includegraphics[width=0.3\linewidth]{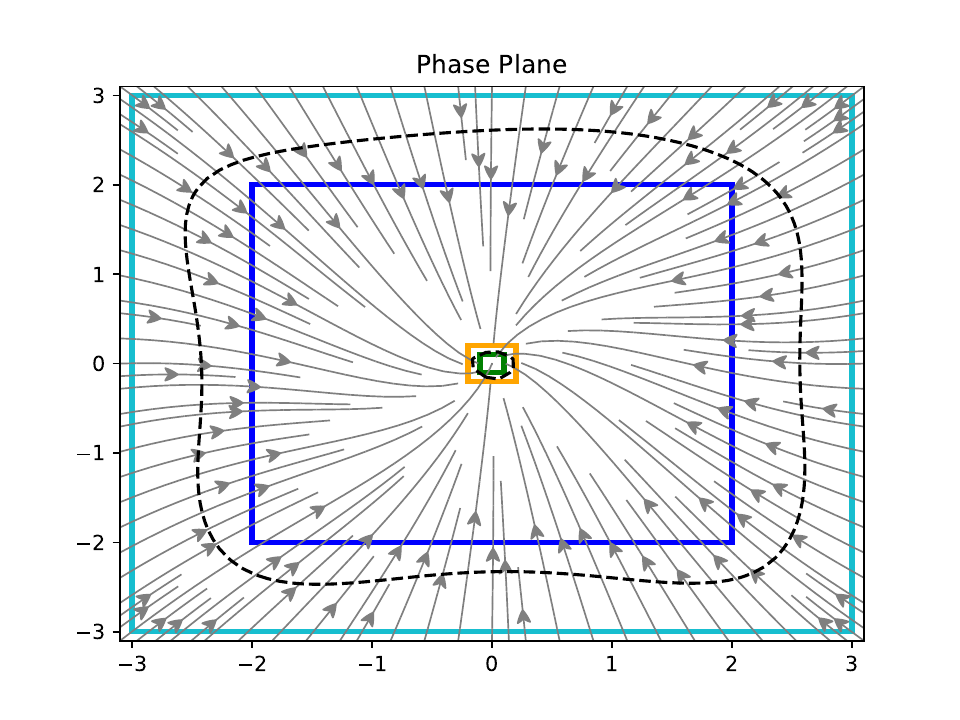} \label{plot:rar}}
     \caption{\new{Visualisations of a selection of certificates as either phase portraits (depicting the dynamics with relevant level sets of the certificate overlaid), or surface plots of the certificates (with relevant sets and level sets shown) from experiments in Table \ref{tab:main}. We show salient level sets as dashed lines, and denote others sets as follows. Dark blue: $\cI$; Red: $\cU$; light blue: $\cS$; green: $\cG$; orange: $\cF$}. %
     }
     \label{plot}
\end{figure*}

\subsection{Control Loss Evaluation}\label{sec:res-cont}
We have presented a novel loss function with the purpose of encouraging trajectories to converge to the origin, which is often desirable in the case of \emph{arrive} conditions, as discussed in Section \ref{sec:cont-loss}
We evaluate the inclusion of this loss function in two ways: first, we consider the effect on success rate and computation time relative to not having the term across all control benchmarks for RWA, RSWA and RAR properties, of which there are 6. These results are presented in Table \ref{tab:loss-comp}, where as before we present the success rate (this time over all 60 runs), and the average computation time for successful runs. It is clear that the inclusion of this loss term improves both the robustness and efficiency of the automated synthesis. 

\begin{table}[htbp]
    \centering
\begin{tabular}{llll}
\toprule
               & Success (\%)          & Time (s)          &  \\
\midrule
With $L_u$    & \textbf{96.67}        & \textbf{7.14 }    &  \\
Without $L_u$ & 80.00                 & 10.64             &  \\
\bottomrule
\end{tabular}

    \caption{Comparison of success rate and computation time for the combined RWA, RSWA and RAR benchmarks with control, with and without the loss term described in \eqref{eq:ctrl-loss}. }
    \label{tab:loss-comp}
\end{table}

Secondly, we compare again these two metrics when instead using an LQR feedback controller. This is the purpose of benchmarks 22 and 24, which are identical except for one key difference: benchmark 22 is equipped with a pre-computed LQR controller, \emph{de facto} representing an autonomous model, whilst we shall compute a control law in benchmark 24.
This allows us to compare our framework's ability to learn feedback laws which satisfy properties relative to a common baseline controller, the known LQR. The results for these two benchmarks are very similar, with our control approach performing slightly more efficiently. Note, we do not claim to outperform LQR controllers, as other cost matrices may perform better or worse; further, we do not provide the cost minimisation guarantees - we simply use these benchmarks as a baseline comparison.

\subsection{Comparison to \textsf{Fossil} 1.0  Baseline}
\label{sec:res-comp}

We have built a prototype tool based on our framework, improving on the work initially presented in~\cite{abate2021FOSSILSoftwareTool}.
This tool compares against competitive state-of-the-art techniques such as SOS-tools, proving to deliver a faster synthesis for stability and safety properties for autonomous models. 
Table~\ref{tab:lit-tab} collects works providing an automated and sound synthesis of  relevant certificates. 
\new{However, these related works either do not have supporting available software or are not maintained, making a direct comparison with these approaches infeasible beyond what has already been presented in presented in~\cite{abate2021FOSSILSoftwareTool}}: 
we therefore benchmark against \textsf{Fossil} 1.0 as the state-of-the-art for certificate synthesis, \new{as an accredited software tool for certificate synthesis that has been previously benchmarked against other approaches, such as SOS tools}. These results are presented in Table \ref{tab:fossil-comp}.

\begin{table}[htbp]
    \centering
    \resizebox{\linewidth}{!}{
    
\begin{tabular}{lllllllll}
    \toprule
                                                                                                 & \rotatebox{90}{Stability} & \rotatebox{90}{ROA} & \rotatebox{90}{Safety} & \rotatebox{90}{SWA} & \rotatebox{90}{RWA} & \rotatebox{90}{RSWA} & \rotatebox{90}{RAR} & \rotatebox{90}{Control}
    \\ \midrule
    {Fossil 2.0}                                                                                 & \cmark                   & \cmark              & \cmark                  & \cmark              & \cmark              & \cmark               & \cmark              & \cmark
    \\
    {Fossil 1.0 \cite{abate2021FOSSILSoftwareTool}}                                              & \cmark                   & \xmark              & \cmark                  & \xmark              & \xmark              & \xmark               & \xmark              & \xmark
    \\
    {F4CS \cite{verdier2020FormalControllerSynthesis},\cite{verdier2020FormalSynthesisAnalytic}} & \cmark                   & \xmark              & \cmark                  & \xmark              & \cmark              & \cmark               & \xmark              & \cmark
    \\
    NLC \cite{chang2020NeuralLyapunovControl}                                                    & \cmark                   & \xmark              & \xmark                  & \xmark              & \xmark              & \xmark               & \xmark              & \cmark
    \\

    \cite{ravanbakhsh2019LearningControlLyapunov, ravanbakhsh2015CounterExampleGuidedSynthesis}, \cite{ravanbakhsh2015CounterexampleGuidedSynthesisReachWS}                                   & \xmark                   & \xmark              & \xmark                  & \xmark              & \cmark              & \cmark               & \xmark              & \cmark
    \\
    \cite{zhao2021SynthesizingReLUNeural}                                                       & \xmark                   & \xmark              & \cmark                  & \xmark              & \xmark              & \xmark               & \xmark              & \xmark \\
    \cite{zhao2020SynthesizingBarrierCertificates}, \cite{zhao2021LearningSafeNeural}            & \xmark                   & \xmark              & \cmark                  & \xmark              & \xmark              & \xmark               & \xmark              & \cmark \\

    \cite{ratschan2018SimulationBasedComputation}                                               & \xmark                   & \xmark              & \cmark                  & \xmark              & \xmark              & \xmark               & \xmark              & \xmark \\
    \cite{kapinski2014SimulationguidedLyapunovAnalysis}                                         & \cmark                   & \xmark              & \cmark                  & \xmark              & \xmark              & \xmark               & \xmark              & \xmark \\ 
    \cite{dai2020CounterexampleGuidedSynthesis}                                                 & \cmark                   & \xmark              & \xmark                  & \xmark              & \xmark              & \xmark               & \xmark              & \xmark \\
    \cite{grande2023augmented, grande2023passsive}                                              & \cmark                   & \xmark              & \xmark                  & \xmark              & \xmark              & \xmark               & \xmark              & \cmark \\
    \bottomrule 
\end{tabular}

    }
    \caption{Comparison of works for automated (and sound) synthesis of certificates for continuous-time dynamical models. We show the properties verified in the respective works, and whether they are also able to verify these properties for control models (either using control certificates or controller and certificate). Similar works are grouped together.
    }
    \label{tab:lit-tab}
\end{table}
We employ the benchmarks originally outlined in \cite{abate2021FOSSILSoftwareTool}, and for a fair comparison we use the same network structure (width and activations) for both tools. It is clear that the approaches are very similar when synthesising the more straightforward Lyapunov functions. However, we achieve significant improvements in terms of both success rate and synthesis time relative to the baseline, on account of a fine-tuned loss function,  
an enhanced communication between the CEGIS components,  
and an overall improved software implementation.

\begin{table*}[ht]
    \centering
    \resizebox{\textwidth}{!}{
    \begin{tabular}{lllllrrrrrrrr}
    \toprule
              &       &             &          &                                                   & \multicolumn{4}{c}{Fosill 1.0} & \multicolumn{4}{c}{Fossil 2.0}                                                \\ \cmidrule(lr){6-9} \cmidrule(lr){10-13}
    Benchmark & $N_s$ & Property & Neurons  & Activations                                       & $\min$                         & $\mu$                          & $\max$ & $S$ & $\min$ & $\mu$ & $\max$ & $S$ \\ 

    \midrule
    NonPoly0  & 2     & Stability   & [5]      & [$\varphi_{2}$]                                  & 0.04                           & 0.21                           & 1.58   & 100 & 0.01   & 0.16  & 1.54   & 100 \\
    Poly2     & 2     & Stability   & [5]      & [$\varphi_{2}$]                                  & 0.35                           & 11.71                          & 70.39  & 90  & 0.08   & 4.77  & 6.50   & 100 \\
    Barr1     & 2     & Safety      & [10]     & [$\varphi_{1}$]                                  & 0.34                           & 1.00                           & 2.72   & 40  & 0.02   & 0.27  & 0.63   & 100 \\
    Barr3     & 2     & Safety      & [10, 10] & [$\sigma_{\textrm{sig}}$,$\sigma_{\textrm{sig}}$] & 16.80                          & 101.72                         & 334.79 & 50  & 3.81   & 14.14 & 30.63  & 100 \\

    \bottomrule
\end{tabular}

    }
    \caption{Comparison of \textsf{Fossil} 1.0 vs \textsf{Fossil} 2.0 (the present work). Here, we use the same naming scheme for benchmarks as used in \textsf{Fossil} 1.0, rather than indexing by number. See Table \ref{tab:main} for details on the columns.}
    \label{tab:fossil-comp}
\end{table*}

\section{Discussions on Generality}
\label{sec:lim}

\subsection{Asymptotic Reachability}
\label{sec:lim-convergence}

Oscillations are important and common phenomena occurring in dynamical systems, typically linked to (stable) limit cycles. 
Certificates for stability analysis in the presence of limit cycles, or equivalently for the asymptotic convergence to such set, have been so far notably absent from the synthesis framework in this manuscript. We discuss them next, 
recalling Barbashin-Krasowskii-Lasalle's Principle \cite{khalil2002NonlinearSystems}.
\begin{theorem}[Invariance Principle \cite{khalil2002NonlinearSystems}]
     Let $\Omega \subset \cX$ be a compact set that is positively invariant with respect to \eqref{eq:dyn-model}. Let $V : \cX \to \mathbb{R}$ be a continuously differentiable function such that $\dot{V}(x) \leq 0$. Let E be the set of all points in $\Omega$ where $V(x)=0$. Let $M$ be the largest invariant set in $E$. Then every solution starting in $\Omega$ approaches $M$ as $t\to \infty$. 
     \hfill$\blacksquare$
\end{theorem}

Crucially, the principle states that 
we can prove asymptotic convergence towards a set thanks to a Lyapunov-like function which is equal to zero \emph{exactly} at the limit cycle. Such certificates have recently been studied from the perspective of disturbed models \cite{meng2021ControlNonlinearSystems, meng2021SmoothConverseLyapunovBarrier}. Nonetheless, without prior knowledge of the existence and location of a limit cycle it is, in our experience, impractical to automatically 
synthesise such certificates, as they would need to be strongly templated based on the limit cycle (namely, precisely tailored to that set). Whilst in principle our approach could offer certificates for such properties, in this work 
we omit certificates requiring such an extensive analysis of the model dynamics.

\subsection{Sufficiency of the Certificates and Completeness of their Synthesis} 
\label{sec:lim-suff}

The certificates provided in this work are \emph{sufficient} proofs for the corresponding properties. We do not in general provide guarantees that a certificate exists if the property is satisfied, i.e. \emph{necessary} proofs.
Such converse results exists for Lyapunov and barrier functions. 
In particular, a construction method of Lyapunov functions is known for 
globally exponentially stable models \cite{khalil2002NonlinearSystems}, 
whilst 
the necessity of barrier certificates is studied in, e.g., \cite{prajna2005necessity,wisniewski2015converse,ratschan2018converse}. 

Contextually, whilst \emph{sound}, our counterexample-based inductive synthesis method is not \emph{complete}: whenever it finds a certificate, the desired property formally holds for the model under consideration. On the other hand, if our algorithm fails to find a certificate we cannot draw any conclusion about the validity of the property for the given model. We show that our procedure consistently terminates with a successful outcome. 

\subsection{Modularity of the Synthesis and Nested Properties} 
We have not considered properties describing that of sequential reachability - namely trajectories arriving at a series of target sets before the goal set. 
We have considered properties obtained by conjunction of requirements, and we shall comment in  \ref{sec:connections} instances obtained via disjunction. We could similarly obtain certificates by manipulating via propositional logic, (e.g., conjunction and disjunction, rather than temporally, as suggested previously) requirements and corresponding certificates. 
Such properties tie into our approach of using certificates in a modular fashion, which is a relatively unexplored concept and represents an area of future work.

\subsection{Issues of Scale}
Our CEGIS-based approach consists of a gradient descent based learning phase followed by an SMT dependent verification phase. While gradient descent is known to scale well to higher dimensions, nonlinear real arithmetic SMT does not in general scale well to higher dimensions. This problem is shared to all methodologies which rely on SMT. %
Possible mitigations, beyond an improvement of SMT performance, include the use of an alternative verification method, e.g., software as Marabou \cite{katz2019MarabouFrameworkVerification} for ReLU networks, or interval bound propagation based techniques.

\subsection{Broader Connections and Taxonomy of Properties}   

In Section \ref{sec:properties} we have presented a diverse set of certifiable properties in terms of the behaviour of trajectories (that is, of solutions) of a given dynamical model. 
Furthermore, in Section \ref{subsec:classification}, we have laid connections across such properties through the notions of \emph{avoid, arrive}, and \emph{remain}. 
In \ref{sec:connections}, we further frame such requirements in broader contexts, providing a categorisation of these  properties within known classes of specifications. 
To this end, we draw connections with formal languages (in particular, regular expressions) and with automata theory (specifically, deterministic finite automata) \cite{hu79}, and in passing we also informally relate to temporal logic for specifications of reactive models~\cite{Pnueli77,bk08,mc18}. 

\ref{sec:connections} is written for the benefit of readers with a background in the mentioned areas, or with an interest in a perspective on dynamical models grounded upon formal methods  \cite{tabuada09,belta2017formal} - this part may be otherwise dispensed with, at no loss of understanding of the overall material.

\section{Concluding Remarks}
\label{sec:conclusion}
We have presented a general framework to formally verify dynamical and control models via certificate synthesis, and introduced \textsf{Fossil} 2.0, a prototype software tool for the automated formal synthesis of a broad range of certificates. Certificate synthesis is based on a CEGIS loop, exploiting neural networks to provide candidate functions, which are then formally verified with the help of SMT solvers. 
Our approach is able to efficiently synthesise certificates for a wide range of properties for both autonomous and control models, with varying complexity of dynamics and sets structure. 
We test our framework and corresponding tool on a number of  benchmarks, outperforming a state-of-the-art tool for certificate synthesis and showing robustness to initialisation. 

In future work, we hope to explore further the modular synthesis of certificates, as well as to extend our framework to stochastic models. 

\begin{table}
    \new{\begin{tabular}{|c | c |}
    \hline
    \textbf{Abbreviation} & \textbf{Definition} \\
    \hline
    ROA & Region of Attraction  \\
    \hline
    RAR & Reach-Avoid-Remain \\
    \hline
    RSWA & Reach-and-Stay While Avoid\\
    \hline
    RWA & Reach-While-Avoid  \\
    \hline
    SWA & Stable While Avoid \\
    \hline
    CEGIS & Counterexample-Guided Inductive \\ 
    & Synthesis\\
    \hline 
    SAT & Satisfiability \\ 
    \hline 
    DPLL &  Davis-Putnam-Logemann-Loveland \\
    \hline
    SMT & Satisfiability Modulo Theories \\
    \hline
    \end{tabular}}
    \caption{List of Abbreviations defined in this work.}
\end{table}

\section*{Acknowledgements}
Alec was supported by the EPSRC Centre for Doctoral Training in Autonomous Intelligent Machines and Systems (EP/S024050/1)

\appendix

\section{Proof of Theorems}
\label{app:proofs}

\begin{proof}(Proof of Thm. \ref{thm:lyap})
We state without proof from Lyapunov theory \cite{sastry1999NonlinearSystems} that the conditions in \eqref{eq:lyap} imply that $f(x)$ is asymptotically stable, and that all sub-level sets of $V(x)$ fully contained within $\cX$ are forward invariant. Since $V$ is continuous and $\cX$ is non-empty, then there exists some $\beta$ such that the sub-level set $\Omega_\beta = \{x \in \cX \ : \ V(x) \leq \beta \}$ is fully contained within $\cX$ and $x^* \in \Omega_\beta$. Since $x^*$ is asymptotically stable, all trajectories in $\Omega_\beta$ converge towards it, and \eqref{eq:stab-spec} holds for the initial set $\Omega_\beta$. 
\end{proof}

\begin{proof}(Proof of Cor. \ref{cor:roa})
Define the $\beta$ sub-level of the ROA certificate $V$ as $\Omega_\beta \coloneqq \{x \in \cX : V(x) \leq \beta \}$. By construction, $\cI \subset \Omega_\beta$, and the condition of  \eqref{eq:lyap} hold. The proof then follows directly from Thm. \ref{thm:lyap}.  
\end{proof}

\begin{theorem}[Nagumo's Theorem] \label{thm:nagumo}
Let us state without proof Nagumo's Theorem. For a proof, see \cite{blanchini2008SetTheoreticMethodsControl}. 
Consider the system $\dot{\xi}(t) =3f(\xi(t))$, where $f$ is Lipschitz continuous such that for each initial condition $\xi(t_0) \in \cX$ it admits a unique solution. Let $\Omega \in \cX$ be a closed set. $\Omega$ is positively invariant if and only if for every exterior normal vector $v$ at point $x$ on the border of $\partial \Omega$, the inner product satisfies $\langle f(x), v \rangle \leq 0$.   
\end{theorem}

\begin{proof}(Proof of Thm. \ref{thm:barrier})
By definition we have that $\cI \cap \cU = \emptyset$. The set $\Omega_0 = \{x \in \cX : \dot{B}(x) \leq 0 \}$ defines a closed set for which \eqref{eq:barr-c} ensures that $f(x)$ points inwards along its border. It follows from \ref{thm:nagumo} that $\Omega_0$ is an forward invariant set which contains $\cI$, and that at all trajectories initialised in $\cI$ remain within $\Omega_0$ for all time $t > t_0$. \eqref{eq:barr-b} ensures that  $\Omega_0 \cap \cU = \emptyset$, and hence \eqref{eq:safety-spec} holds. 
\end{proof}

\begin{proof}(Proof of Cor. \ref{cor:stab-safe})
We note that the specification described by \eqref{eq:stab-safe-spec} is the conjunction of those in  \eqref{eq:stab-spec} and  \eqref{eq:safety-spec}. Therefore, the proof follows directly from Cor. \ref{cor:roa} and Thm. \ref{thm:barrier}. 
\end{proof}

\begin{proof}(Proof of Thm. \ref{thm:rwa} (from \cite{verdier2020FormalControllerSynthesis}))
Let $A = \{ x \in \cS : V(x) \leq 0\}$.
For $\xi(t_0)\in \cI$, it follows from \eqref{eq:rws-cert-a} and the definition of $A$ that $\xi(t_0) \in A$. From \eqref{eq:rws-cert-c}, for all $\xi(t_k) \in A \setminus \cG$, $\dot{V}(\xi(t_k)) < -\gamma$. Using $\forall x \in A, V(x) \leq 0$ and the comparison principle \cite{khalil2002NonlinearSystems}, it follows that $\forall k \in \mathbb{Z}_{\geq 0}, \forall t \in [t_k, t_k + h], \forall \xi(t_k) \in A \setminus \cG: V(\xi(t)) \leq V(\xi(t_k)) - \gamma h \leq -\gamma h$. Therefore, $\xi(t_k) \in A \setminus \cG$ implies $\forall t \in [t_k, t_k+h], V(\xi(t)$ will decrease and thus cannot reach $\partial \cS$, since by \eqref{eq:rws-cert-b} $\forall x \in \partial \cS: V(x) > 0$. Since $\cS$ is compact and $V(x)$ is continuous, $\exists e \in \real$ s.t. $e = \inf_{x \in \cS \setminus \cG} V(x)$ and the sublevel sets of $V(x)$ is compact. It follows that $V(x)$ is lower bounded on $A \setminus \cG \subset \cS \setminus \cG$, so $V(\xi(t))$ will decrease until in finite time $\xi(t)$ leaves $A \setminus \cG$ and may only enter $\cG$. 
\end{proof}

\begin{proof}(Proof of Thm. \ref{thm:rswa} (from \cite{verdier2020FormalControllerSynthesis})) 
   Let $B = \{ x \in \cS : V(x) \leq \beta\}$. From Thm. \ref{thm:rwa}, there exists a time $T \geq t_0$ such that $\xi(T) \in \cF$.
   Using a similar argument as before, we conclude that $\forall \xi(T) \in \cF, \xi(t)$ with $t\geq T$ enters in finite time $\cF \cap B$. Since $B$ is a sub-level set of the compact set $S$ and continuous $V$, then $B$ is also compact and so is $\cF \cap B$
   . From \eqref{eq:rswa-cert-e} we have $\forall x \in \partial(\cF \cap B): \dot{V}(x) \leq -\gamma$. Combining with \eqref{eq:rswa-cert-d}, we have that all states $\xi(t) \in \partial(\cF \cap B)$ cannot reach $\partial \cF$ and $V(\xi(t))$ decreases, meaning these trajectories remain in $\cF \cap B$. Therefore, $\cF \cap B$ is forward invariant. Since $\cF \subset \interior(\cS)$, we have that \eqref{eq:rswa-spec} holds. 
\end{proof}

\begin{proof}(Proof of \ref{thm:rar})
This proof follows directly from Thms. \ref{thm:barrier} and \ref{thm:rwa}.     
\end{proof}

\section{Broader Connections and Taxonomy of Properties}  
\label{sec:connections}

In this part, we further frame the presented  requirements in broader contexts, thus providing a categorisation of the properties under study within specific classes of \emph{formal specifications} \cite{bk08,mc18}. In order to do so, we draw connections with formal languages (in particular, regular expressions) and with automata theory (specifically, deterministic finite automata) \cite{hu79}, and in passing we also informally relate to the use of temporal logic for specifications of reactive models  \cite{Pnueli77,bk08,mc18}. 
However, these connections can be drawn under an important proviso:  that is, in formal methods, properties by and large are expressed over finite alphabets, namely over a finite set of labels (which can be related to our sets -- or complement thereof -- of interest), but most importantly concern traces in discrete time. Similarly, formal languages and automata deal with finite or countably-infinite strings of (finite) characters.

\paragraph*{Safety and reachability - language and semantic duality} 
We remark that all the presented properties consist of combinations of two fundamental specifications in formal verification \cite{10.1145/93385.93442,KV99}. The first is \emph{safety}, which qualitatively concerns  trajectories always staying clear from a nominative region in the state space that is deemed to be unsafe. Thus, safety specifications are infinite-horizon requirements, which however can be equivalently defined as requirements admitting finite-horizon counter-examples: namely, they are specifications that are necessarily invalidated by trajectories that enter the given unsafe in finite time. 

The second specification we refer to as \emph{reachability}, and comprise trajectories reaching a given desirable goal set (which is also known as \emph{reach} or \emph{target} set). 
In the area of formal languages, reachability is known to be a \emph{co-safety} property, namely a dual to a safety requirement and, as such, it is a finite-horizon property. Indeed, in formal verification co-safe specifications admit finite-horizon witnesses, that is satisfying trajectories that enter the goal set in finite time.  
In Section \ref{sec:properties} we have defined this as unconstrained reachability, and as a special instance of the reach-while avoid specification. We could have alternatively introduced it as the logical dual of safety, except that the certificate synthesis would not have followed as seamlessly. 
We should mention that, more generally, co-safety properties (and, in particular, reachability) are special instances of another broad class of specifications, known as \emph{liveness} properties \cite{bk08,mc18}: in this work we do not deal with general liveness properties, and in particular the synthesis of sufficient certificates for this class is left as future work.  

\paragraph*{Finite- and infinite-horizon requirements} 
We have discussed that safety and reachability properties are dual in the sense that the earlier raises a requirement over infinite-horizon trajectories, asking that \emph{nothing bad ever happens}, whereas the latter requires reasoning about finite-horizon solutions, asking that over a finite time horizon \emph{something good eventually happens} \cite{bk08}.   
However note that, whenever safety requirements are added to finite-time reachability properties, as in the case of the \emph{avoid} constraint in RWA, RSWA, and RAR, these safety requirements ought to hold only over the finite time spans inherited from the reachability requirements. 

Dually, as much as reachability requirements natively encompass co-safety properties, in this work we have discussed extension of such reachability requirements over \emph{infinite-time} horizons: this is quite natural in control theory, namely in the context of asymptotic stability.  
Similarly, we have not only raised requirements that trajectories reach a goal set within a finite time horizon, which we have simply denoted as \emph{reachability}, 
but additionally discussed (cf. \ref{sec:lim-convergence}) that they may approach the set (e.g., a set of equilibria, or a limit cycle) asymptotically, which we have denoted as \emph{asymptotic reachability}. 

\paragraph*{Automata theory} 
Next, we qualitatively relate these classes of properties to automata theory. 
Through their duality, both safety and co-safety properties can be expressed by the same class of finite-state models, namely deterministic finite automata (DFA) \cite{hu79}. In other words, traces within these two classes of specifications can be compiled by a DFA model, which reads them when they are started in one of its initial states, and when they terminate upon hitting one its accepting states. (Conversely,  liveness properties, which we saw generalise co-safety specifications, require automata of different nature, such as B\"uchi or Rabin automata, which specifically accept infinite traces.) 
Summarising this discussion, we can conclude that we can provide sufficient certificates for properties that can be expressed as DFAs: 
an enticing extension of our work is looking at richer specifications obtained by modularly composing such finite-state models. 

\paragraph*{Linear-time properties and temporal logic}
Finally, let us draw connections to temporal logic, a class of modal logic that was natively developed to specify requirements of (models of) reactive programs  \cite{Pnueli77}. Whilst originally employed for finite-state programs evolving over discrete-time steps, temporal logics, such as LTL (linear TL), have been also widely employed in the context of dynamical and control models \cite{tabuada09,belta2017formal}.  
Bearing in mind that the above caveat on discrete-time semantics holds also in this context (which for instance renders its \emph{next} operator useless for our purposes),  we can express safety requirements over safe set $S$ as \emph{always} $S$ - in LTL syntax this is expressed as the formula $\Box S$ (or $\mathsf G S$) - and, dually, reachability requirements over target set $T$ as \emph{eventually} $T$ - in LTL this is $\diamond T$ (or $\mathsf F T$); incidentally, the discussed duality between these requirements is crisply expressed logically as: $\Box S \equiv \neg \diamond \neg S \equiv \neg \diamond S^\complement$. Similar discussions follow through for finite-time properties, which in LTL require selecting a (finite) horizon $\tau \in \real$, yielding for example $\Box^{\leq \tau} S$. Reach-avoid can be expressed via the \emph{until} operator $\mathsf U$ in LTL, say $S\, \mathsf U T$, which may be tailored to unconstrained reachability as $\diamond T \equiv \mathtt{true} \, \mathsf U T$.  

Dovetailing to the discussion above, we can thus flexibly and modularly synthesise certificates also for the known \emph{weak until} $\mathsf W$, and even for the \emph{release} $\mathsf R$ specifications in LTL. Indeed, recall that \cite{bk08}   
$$
S \, \mathsf W T = (S \, \mathsf U T) \vee \Box S,   
$$
We can thus obtain a certificate for weak until from either a certificate for reach-avoid or one for safety, respectively. Additionally notice in particular, that $\Box S \equiv S \, \mathsf W \mathtt{false}$.

Finally, the useful \emph{release} operator in LTL expresses a safety requirement (set $S$) that is released upon reaching a target set ($T$), as follows:  
$$
T \, \mathsf R\, S  = \neg (\neg T \, \mathsf U \neg S),    
$$ 
and in fact (the next equivalent expression is easier for finding certificates)  %
$$
T\, \mathsf R\, S  = (\neg T \wedge S) \mathsf W (T \wedge S),    
$$ 
and so, in particular, $\Box S \equiv \mathtt{false} \, \mathsf R S$.

 \bibliographystyle{elsarticle-num} 
 \bibliography{main}

\pagebreak
 \onecolumn
\section{Benchmarks for Certificates Chapter}
\label{app:benchmarks_certificates}
Here, we detail the benchmarks presented in this work. These benchmarks are indexed with a number, corresponding to the row index from Table \ref{tab:main}, along with the property that is being verified for the benchmark.
We show the salient sets for the property; note that in some aligned it is simpler to show the safe set $\mathcal{X}_S$, and in other aligned the unsafe set $\mathcal{X}_U$. These are always the complement of each other: $\mathcal{X}_S = \mathcal{X} \setminus \mathcal{X}_U$.

We also use shorthand notation to describe sets of familiar geometric structure:
\begin{itemize}
    \item $\textrm{Rectangle}(lb, ub)$ denotes a hyper rectangle characterised by its lower bounds ($lb$) for each dimension, and its upper bounds ($ub$) in each dimension.
    \item $\textrm{Sphere}(c, r)$ denotes a hyper sphere with centre $c$ and radius $r$. The dimension of the hyper-sphere corresponds to the size of $c$.
    \item $\textrm{Torus}(c, r_i, r_o)$ denotes the set defined by $\textrm{Sphere}(c, r_o) \setminus \textrm{Sphere}(c, r_i)$; we note that this is not a true torus, but rather a sphere with a spherical hole in the middle.
\end{itemize}

\subsubsection*{ [1] NonPoly0 Stability }

$$f(x) = \left[ x_{0} x_{1} - x_{0}, \  - x_{1}\right]$$

$$\mathcal{X}: \Torus([0, 0], 1, 0.01) $$

Activation: [$\varphi_{2}$], Neurons: [6]

\subsubsection*{ [2] Poly1 Stability }

$$f(x) = \left[ - x_{0}^{3} - x_{0} x_{2}^{2}, \  - x_{0}^{2} x_{1} - x_{1}, \  3 x_{0}^{2} x_{2} - 4 x_{2}\right]$$

$$\mathcal{X}: \Torus([0.0, 0.0, 0.0], 10.0, 0.1) $$

Activation: [$\varphi_{2}$], Neurons: [8]

\subsubsection*{ [3] Benchmark1 Stability }

$$f(x) = \left[ u_{0} + x_{0} + x_{1}, \  u_{1} - x_{0} - x_{1}\right]$$

$$\mathcal{X}: \Torus([0.0, 0.0], 10.0, 0.1) $$

Activation: [$\varphi_{2}$], Neurons: [4]
CTRLActivation: [$\varphi_{1}$], Neurons: [15, 2]

\subsubsection*{ [4] InvertedPendulum Stability }

$$f(x) = \left[ u_{0} + x_{1}, \  u_{1} - \frac{8}{3} x_{1} + 19.62 \sin{\left(x_{0} \right)}\right]$$

$$\mathcal{X}: \Torus([0.0, 0.0], 1, 0.1) $$

Activation: [$\varphi_{2}$], Neurons: [5]
CTRLActivation: [$\varphi_{1}$], Neurons: [25, 2]

\subsubsection*{ [5] NonPoly1 ROA }

$$f(x) = \left[ 2 x_{0}^{2} x_{1} - x_{0}, \  - x_{1}\right]$$

$$\mathcal{X}_I: (\Sphere([-1, 1], 0.1) | \Sphere([1, -1], 0.2)) $$

Activation: [$\sigma_{\mathrm{t}^2}$], Neurons: [5]

\subsubsection*{ [6] LorenzSystem ROA }

$$f(x) = \left[ u_{0} - 10.0 x_{0} + 10.0 x_{1}, \  u_{1} - x_{0} x_{2} + 28.0 x_{0} - x_{1}, \  u_{2} + x_{0} x_{1} - \frac{8}{3} x_{2}\right]$$

$$\mathcal{X}_I: \Sphere([0, 0, 0], 0.3) $$

Activation: [$\varphi_{2}$], Neurons: [8]
CTRLActivation: [$\varphi_{1}$], Neurons: [8, 3]

\subsubsection*{ [7] Barr2 Safety }

$$f(x) = \left[ x_{1} - 1 + e^{- x_{0}}, \  - \sin^{2}{\left(x_{0} \right)}\right]$$

$$\mathcal{X}: x_{0} \geq -2 \wedge x_{1} \leq 2 $$

$$\mathcal{X}_I: \Sphere([-0.5, 0.5], 0.4) $$

$$\mathcal{X}_U: \Sphere([0.7, -0.7], 0.3) $$

Activation: [$\sigma_{\mathrm{t}}$], Neurons: [15]

\subsubsection*{ [8] ObstacleAvoidance Safety }

$$f(x) = \left[ \sin{\left(x_{2} \right)}, \  \cos{\left(x_{2} \right)}, \  \frac{3 x_{0} \sin{\left(x_{2} \right)} + 3 x_{1} \cos{\left(x_{2} \right)}}{x_{0}^{2} + x_{1}^{2} + 0.5} - \sin{\left(x_{2} \right)}\right]$$

$$\mathcal{X}: x_{0} \geq -2 \wedge x_{1} \geq -2 \wedge x_{2} \geq -1.57 \wedge x_{0} \leq 2 \wedge x_{1} \leq 2 \wedge x_{2} \leq 1.57 $$

$$\mathcal{X}_I: x_{0} \geq -0.1 \wedge x_{1} \geq -2 \wedge x_{2} \geq -0.52 \wedge x_{0} \leq 0.1 \wedge x_{1} \leq -1.8 \wedge x_{2} \leq 0.52 $$

$$\mathcal{X}_U: x_{0}^{2} + x_{1}^{2} \leq 0.04 $$

Activation: [$\varphi_{4}$], Neurons: [25]

\subsubsection*{ [9] HighOrd8 Safety }

\begin{align*}
    f(x) = [ & x_{1}, \  x_{2}, \  x_{3}, \  x_{4}, \  x_{5}, \  x_{6}, \  x_{7},                                  \\
             & - 576 x_{0} - 2400 x_{1} - 4180 x_{2} - 3980 x_{3} - 2273 x_{4} - 800 x_{5} - 170 x_{6} - 20 x_{7}]
\end{align*}

$$\mathcal{X}: \Rectangle([-2.2, -2.2, -2.2, -2.2, -2.2, -2.2, -2.2, -2.2], [2.2, 2.2, 2.2, 2.2, 2.2, 2.2, 2.2, 2.2]) $$

$$\mathcal{X}_I: \Rectangle([0.9, 0.9, 0.9, 0.9, 0.9, 0.9, 0.9, 0.9], [1.1, 1.1, 1.1, 1.1, 1.1, 1.1, 1.1, 1.1]) $$

\begin{align*}
    \mathcal{X}_U: \Rectangle( & [-2.2, -2.2, -2.2, -2.2, -2.2, -2.2, -2.2, -2.2], \\ &[-1.8, -1.8, -1.8, -1.8, -1.8, -1.8, -1.8, -1.8])
\end{align*}

Activation: [$\varphi_{1}$], Neurons: [10]

\subsubsection*{ [10] CtrlObstacleAvoidance Safety }

$$f(x) = \left[ \sin{\left(x_{2} \right)}, \  \cos{\left(x_{2} \right)}, \  u_{0} - \sin{\left(x_{2} \right)}\right]$$

$$\mathcal{X}: \Rectangle([-10.0, -10.0, -\pi], [10.0, 10.0, \pi]) $$

$$\mathcal{X}_I: \Rectangle([4.0, 4.0, \frac{-\pi}{2}], [6.0, 6.0, \frac{\pi}{2}]) $$

$$\mathcal{X}_U: \Rectangle([-9, -9, \frac{-\pi}{2}], [-7.0, -6.0, \frac{\pi}{2}]) $$

Activation: [$\sigma_{\mathrm{t}}$], Neurons: [15]
CTRLActivation: [$\sigma_{\mathrm{t}}$], Neurons: [5, 1]

\subsubsection*{ [11] NonPoly3 SWA }

$$f(x) = \left[ - 0.1 x_{0} x_{1}^{3} - 3 x_{0}, \  - x_{1} + x_{2}, \  - x_{2}\right]$$

$$\mathcal{X}: \Torus([0, 0, 0], 3, 0.01) $$

$$\mathcal{X}_I: \Sphere([-0.9, -0.9, -0.9], 1.0) $$

$$\mathcal{X}_U: (\Sphere([0.4, 0.4, 0.4], 0.2) | \Sphere([-0.4, 0.4, 0.4], 0.2)) $$

Activation: [$\varphi_{2}$], Neurons: [6]
AltActivation: [$\sigma_{\mathrm{t}}$], AltNeurons: [5]

\subsubsection*{ [12] Barr3 SWA }

$$f(x) = \left[ x_{1}, \  \frac{1}{3} x_{0}^{3} - x_{0} - x_{1}\right]$$

$$\mathcal{X}: \Rectangle([-3, -2], [2.5, 1]) $$

$$\mathcal{X}_I: \Rectangle([0.4, 0.1], [0.8, 0.5]) $$

$$\mathcal{X}_U: \Sphere([-1, -1], 0.4) $$

Activation: [$\varphi_{2}$], Neurons: [5]
AltActivation: [$\sigma_{\mathrm{sig}}$,$\varphi_{2}$], AltNeurons: [5, 5]

\subsubsection*{ [13] SecondOrder SWA }

$$f(x) = \left[ - x_{0}^{3} + x_{1}, \  u_{0}\right]$$

$$\mathcal{X}: \Rectangle([-1.5, -1.5], [1.5, 1.5]) $$

$$\mathcal{X}_I: \Rectangle([-0.5, -0.5], [0.5, 0.5]) $$

$$\mathcal{X}_U: Complement(\Rectangle([-1, -1], [1, 1])) $$

Activation: [$\varphi_{2}$], Neurons: [8]
AltActivation: [$\varphi_{2}$], AltNeurons: [5]
CTRLActivation: [$\varphi_{1}$], Neurons: [8, 1]

\subsubsection*{ [14] ThirdOrder SWA }

$$f(x) = \left[ u_{0} - 10 x_{0} + 10 x_{1}, \  - x_{0} x_{2} + 28 x_{0} - x_{1}, \  x_{0} x_{1} - \frac{8}{3} x_{2}\right]$$

$$\mathcal{X}: \Rectangle([-6, -6, -6], [6, 6, 6]) $$

$$\mathcal{X}_I: \Rectangle([-1.2, -1.2, -1.2], [1.2, 1.2, 1.2]) $$

$$\mathcal{X}_U: Complement(\Rectangle([-5, -5, -5], [5, 5, 5])) $$

Activation: [$\varphi_{2}$], Neurons: [10]
AltActivation: [$\sigma_{\mathrm{t}}$], AltNeurons: [8]
CTRLActivation: [$\varphi_{1}$], Neurons: [8, 1]

\subsubsection*{ [15] SecondOrderLQR RWA }

$$f(x) = \left[ - x_{0}^{3} + x_{1}, \  - 1.0 x_{0} - 1.73 x_{1}\right]$$

$$\mathcal{X}: \Rectangle([-1.5, -1.5], [1.5, 1.5]) $$

$$\mathcal{X}_I: \Rectangle([-0.5, -0.5], [0.5, 0.5]) $$

$$\mathcal{X}_S: \Rectangle([-1, -1], [1, 1]) $$

$$\mathcal{X}_G: \Rectangle([-0.05, -0.05], [0.05, 0.05]) $$

Activation: [$\varphi_{2}$], Neurons: [4]

\subsubsection*{ [16] ThirdOrderLQR RWA }

$$f(x) = \left[ - 33.71 x_{0} - 8.49 x_{1}, \  - x_{0} x_{2} + 28 x_{0} - x_{1}, \  x_{0} x_{1} - \frac{8}{3} x_{2}\right]$$

$$\mathcal{X}: \Rectangle([-6, -6, -6], [6, 6, 6]) $$

$$\mathcal{X}_I: \Rectangle([-1.2, -1.2, -1.2], [1.2, 1.2, 1.2]) $$

$$\mathcal{X}_S: \Rectangle([-5, -5, -5], [5, 5, 5]) $$

$$\mathcal{X}_G: \Rectangle([-0.3, -0.3, -0.3], [0.3, 0.3, 0.3]) $$

Activation: [$\varphi_{2}$], Neurons: [16]

\subsubsection*{ [17] SecondOrder RWA }

$$f(x) = \left[ - x_{0}^{3} + x_{1}, \  u_{0}\right]$$

$$\mathcal{X}: \Rectangle([-1.5, -1.5], [1.5, 1.5]) $$

$$\mathcal{X}_I: \Rectangle([-0.5, -0.5], [-0.1, -0.1]) $$

$$\mathcal{X}_S: (\Rectangle([-1.5, -1.5], [1.5, 1.5]) \setminus \Sphere([0.5, 0.5], 0.2)) $$

$$\mathcal{X}_G: \Rectangle([-0.05, -0.05], [0.05, 0.05]) $$

Activation: [$\sigma_{\mathrm{sig}}$,$\varphi_{2}$], Neurons: [4, 4]
CTRLActivation: [$\varphi_{1}$], Neurons: [8, 1]

\subsubsection*{ [18] ThirdOrder RWA }

$$f(x) = \left[ u_{0} - 10 x_{0} + 10 x_{1}, \  - x_{0} x_{2} + 28 x_{0} - x_{1}, \  x_{0} x_{1} - \frac{8}{3} x_{2}\right]$$

$$\mathcal{X}: \Rectangle([-6, -6, -6], [6, 6, 6]) $$

$$\mathcal{X}_I: \Rectangle([-1.2, -1.2, -1.2], [1.2, 1.2, 1.2]) $$

$$\mathcal{X}_S: \Rectangle([-5, -5, -5], [5, 5, 5]) $$

$$\mathcal{X}_G: \Rectangle([-0.3, -0.3, -0.3], [0.3, 0.3, 0.3]) $$

Activation: [$\varphi_{2}$], Neurons: [5]
CTRLActivation: [$\varphi_{1}$], Neurons: [8, 1]

\subsubsection*{ [19] InvertedPendulum RWA }

$$f(x) = \left[ u_{0} + x_{1}, \  u_{1} - \frac{8}{3} x_{1} + 19.62 \sin{\left(x_{0} \right)}\right]$$

$$\mathcal{X}: \Rectangle([-3, -3], [3, 3]) $$

$$\mathcal{X}_I: \Rectangle([-0.6, -0.6], [0.6, 0.6]) $$

$$\mathcal{X}_S: \Rectangle([-2.5, -2.5], [2.5, 2.5]) $$

$$\mathcal{X}_G: \Rectangle([-0.01, -0.01], [0.01, 0.01]) $$

Activation: [$\sigma_{\mathrm{sig}}$], Neurons: [5]
CTRLActivation: [$\varphi_{1}$], Neurons: [8, 2]

\subsubsection*{ [20] SecondOrderLQR RSWA }

$$f(x) = \left[ - x_{0}^{3} + x_{1}, \  - 1.0 x_{0} - 1.73 x_{1}\right]$$

$$\mathcal{X}: \Rectangle([-1.5, -1.5], [1.5, 1.5]) $$

$$\mathcal{X}_I: \Rectangle([-0.5, -0.5], [0.5, 0.5]) $$

$$\mathcal{X}_S: \Rectangle([-1, -1], [1, 1]) $$

$$\mathcal{X}_F: \Rectangle([-0.05, -0.05], [0.05, 0.05]) $$

Activation: [$\varphi_{2}$], Neurons: [4]

\subsubsection*{ [21] ThirdOrderLQR RSWA }

$$f(x) = \left[ - 33.71 x_{0} - 8.49 x_{1}, \  - x_{0} x_{2} + 28 x_{0} - x_{1}, \  x_{0} x_{1} - \frac{8}{3} x_{2}\right]$$

$$\mathcal{X}: \Rectangle([-6, -6, -6], [6, 6, 6]) $$

$$\mathcal{X}_I: \Rectangle([-1.2, -1.2, -1.2], [1.2, 1.2, 1.2]) $$

$$\mathcal{X}_S: \Rectangle([-5, -5, -5], [5, 5, 5]) $$

$$\mathcal{X}_F: \Rectangle([-0.3, -0.3, -0.3], [0.3, 0.3, 0.3]) $$

Activation: [$\varphi_{2}$], Neurons: [16]

\subsubsection*{ [22] InvertedPendulumLQR RSWA }

$$f(x) = \left[ - 7.21 x_{0} - 0.34 x_{1}, \  - 1.34 x_{0} - 2.997 x_{1} + 19.62 \sin{\left(x_{0} \right)}\right]$$

$$\mathcal{X}: \Rectangle([-3, -3], [3, 3]) $$

$$\mathcal{X}_I: \Rectangle([-0.6, -0.6], [0.6, 0.6]) $$

$$\mathcal{X}_S: \Rectangle([-2.5, -2.5], [2.5, 2.5]) $$

$$\mathcal{X}_F: \Rectangle([-0.3, -0.3], [0.3, 0.3]) $$

Activation: [$\sigma_{\mathrm{sig}}$,$\varphi_{2}$], Neurons: [5, 5]

\subsubsection*{ [23] SecondOrder RSWA }

$$f(x) = \left[ - x_{0}^{3} + x_{1}, \  u_{0}\right]$$

$$\mathcal{X}: \Rectangle([-1.5, -1.5], [1.5, 1.5]) $$

$$\mathcal{X}_I: \Rectangle([-0.5, -0.5], [0.5, 0.5]) $$

$$\mathcal{X}_S: \Rectangle([-1, -1], [1, 1]) $$

$$\mathcal{X}_F: \Rectangle([-0.05, -0.05], [0.05, 0.05]) $$

Activation: [$\varphi_{2}$], Neurons: [8]
CTRLActivation: [$\varphi_{1}$], Neurons: [8, 1]

\subsubsection*{ [24] InvertedPendulum RSWA }

$$f(x) = \left[ u_{0} + x_{1}, \  u_{1} - \frac{8}{3} x_{1} + 19.62 \sin{\left(x_{0} \right)}\right]$$

$$\mathcal{X}: \Rectangle([-3, -3], [3, 3]) $$

$$\mathcal{X}_I: \Rectangle([-0.6, -0.6], [0.6, 0.6]) $$

$$\mathcal{X}_S: \Rectangle([-2.5, -2.5], [2.5, 2.5]) $$

$$\mathcal{X}_F: \Rectangle([-0.3, -0.3], [0.3, 0.3]) $$

Activation: [$\sigma_{\mathrm{sig}}$,$\varphi_{2}$], Neurons: [5, 5]
CTRLActivation: [$\varphi_{1}$], Neurons: [8, 2]

\subsubsection*{ [25] SecondOrderLQR RAR }

$$f(x) = \left[ - x_{0}^{3} + x_{1}, \  - 1.0 x_{0} - 1.73 x_{1}\right]$$

$$\mathcal{X}: \Rectangle([-3.5, -3.5], [3.5, 3.5]) $$

$$\mathcal{X}_I: \Rectangle([-2, -2], [2, 2]) $$

$$\mathcal{X}_S: \Rectangle([-3, -3], [3, 3]) $$

$$\mathcal{X}_G: \Rectangle([-0.1, -0.1], [0.1, 0.1]) $$

$$\mathcal{X}_F: \Rectangle([-0.15, -0.15], [0.15, 0.15]) $$

Activation: [$\sigma_{\mathrm{soft}}$], Neurons: [6]
AltActivation: [$\varphi_{2}$], AltNeurons: [6]

\subsubsection*{ [26] InvertedPendulum RAR }

$$f(x) = \left[ u_{0} + x_{1}, \  u_{1} - \frac{8}{3} x_{1} + 19.62 \sin{\left(x_{0} \right)}\right]$$

$$\mathcal{X}: \Rectangle([-3.5, -3.5], [3.5, 3.5]) $$

$$\mathcal{X}_I: \Rectangle([-2, -2], [2, 2]) $$

$$\mathcal{X}_S: \Rectangle([-3, -3], [3, 3]) $$

$$\mathcal{X}_G: \Rectangle([-0.1, -0.1], [0.1, 0.1]) $$

$$\mathcal{X}_F: \Rectangle([-0.2, -0.2], [0.2, 0.2]) $$

Activation: [$\sigma_{\mathrm{sig}}$,$\varphi_{2}$], Neurons: [6, 6]
AltActivation: [$\sigma_{\mathrm{sig}}$,$\varphi_{2}$], AltNeurons: [6, 6]
CTRLActivation: [$\varphi_{1}$], Neurons: [8, 2]

\twocolumn

\end{document}